%% file: minimizing.tex
\PassOptionsToPackage{final}{graphicx}
\documentclass[final,UKenglish,a4paper,thm-restate,autoref,numberwithinsect]{lipics-v2021}

\pdfoutput=1
\nolinenumbers                  %

\usepackage{amssymb,amsmath,mathrsfs}
\usepackage{mathtools}
\usepackage{thmtools}
\usepackage{xspace}
\usepackage{dsfont}
\usepackage{tikz}
\usetikzlibrary{cd,fit,calc,positioning,arrows}
\usepackage{ifthen}
\usepackage[notref,notcite]{showkeys}
\usepackage[noadjust]{cite}
\usepackage{rotating}

\usepackage{ifdraft}
\ifdraft{
  \overfullrule=0mm %
}{}

\newcommand{\resetCurThmBraces}{%
\gdef\curThmBraceOpen{(}%
\gdef\curThmBraceClose{)}}
\resetCurThmBraces
\newcommand{\removeThmBraces}{%
\gdef\curThmBraceOpen{}%
\gdef\curThmBraceClose{}}
\resetCurThmBraces

\newenvironment{notheorembrackets}{\removeThmBraces}{\resetCurThmBraces}

\usepackage{etoolbox}
\patchcmd{\thmhead}{(#3)}{\curThmBraceOpen #3\curThmBraceClose }{}{}

\usepackage{hyperref}
\hypersetup{hidelinks,final,bookmarks}

\usepackage{seqsplit}
\usepackage{xstring}
\newcommand{\defaultshowkeysformat}[1]{%
\StrSubstitute{#1}{ }{\textvisiblespace}[\TEMP]%
\parbox[t]{\marginparwidth}{\raggedright\normalfont\small\ttfamily\(\{\){\color{red!50!black}\expandafter\seqsplit\expandafter{\TEMP}}\(\}\)}%
}

\renewcommand*\showkeyslabelformat[1]{%
\noexpandarg%
\defaultshowkeysformat{#1}%
}

\newcommand{\itemref}[2]{\autoref{#1}.\ref{#2}}

\usepackage[footnote,marginclue,nomargin]{fixme}

\newcommand{\xra}[1]{\xrightarrow{~#1~}}

\newcommand{\CO}{\mathcal{O}}
\newcommand{\op}[1]{\operatorname{\mathsf{#1}}}
\newcommand{\id}{\op{id}}
\newcommand{\fil}{\op{fil}}
\newcommand{\fils}[1]{\fil_{\{#1\}}} %
\newcommand{\opcit}[1][.]{\textit{op.cit#1}\xspace}
\newcommand{\group}{\op{group}}
\newcommand{\ungroup}{\op{ungroup}}
\newcommand{\curry}{\op{curry}}
\newcommand{\uncurry}{\op{uncurry}}
\newcommand{\swap}{\op{swap}}
\newcommand{\ev}{\op{ev}}
\newcommand{\inj}{\op{in}}
\newcommand{\inl}{\op{inl}}
\newcommand{\inr}{\op{inr}}
\newcommand{\pr}{\op{pr}}
\newcommand{\ar}{\op{ar}}
\newcommand{\monoto}{\ensuremath{\rightarrowtail}}
\newcommand{\epito}{\ensuremath{\twoheadrightarrow}}
\newcommand{\subto}{\ensuremath{\hookrightarrow}}
\newcommand{\filter}{\op{filter}}
\newcommand{\init}{\ensuremath{\mathtt{init}}\xspace}
\newcommand{\update}{\ensuremath{\mathtt{update}}\xspace}
\newcommand{\merge}{\ensuremath{\mathtt{merge}}}
\newcommand{\concat}{\ensuremath{\mathtt{concat}}}
\newcommand{\Set}{\mathsf{Set}}
\newcommand{\Pow}{\mathcal{P}}
\newcommand{\Dist}{\mathcal{D}}
\newcommand{\Powf}{\mathcal{P}_\mathsf{f}} %
\newcommand{\Bag}{\mathcal{B}}
\newcommand{\copar}{CoPaR\xspace}
\newcommand{\N}{\mathds{N}}
\newcommand{\R}{\mathds{R}}

\newcommand{\fpair}[1]{\ensuremath{\langle #1 \rangle}}
\newcommand{\prodi}[1][i]{{\textstyle\prod_{#1\in I}}}
\newcommand{\coprodi}[1][i]{{\textstyle\coprod_{#1\in I}}}

\newcommand{\lmbrace}{\{\mskip-4mu[}
\newcommand{\rmbrace}{]\mskip-4mu\}}
\newcommand{\mbracesSpace}[1]{\ifthenelse{\equal{#1}{}}{}{\,}}
\newcommand{\mbraces}[1]{%
  \lmbrace{}\protect\mbracesSpace{#1}#1\protect\mbracesSpace{#1}\rmbrace{}}
\newcommand{\emptybag}{\lmbrace{}\rmbrace{}}
\newcommand{\Coalg}{\mathop{\mathsf{Coalg}}}

\newcommand{\fullorax}{appendix}

\newcommand{\takeout}[1]{\empty}

\newcommand{\descto}[3][]{\arrow[phantom]{#2}[#1]{\text{\footnotesize{}\begin{tabular}{c}#3\end{tabular}}}}
\newcommand{\desctox}[4][]{\arrow[phantom,#2]{#3}[#1]{\text{\footnotesize{}\begin{tabular}{c}#4\end{tabular}}}}

\tikzset{shiftarr/.style={
        rounded corners,%
        to path={--([#1]\tikztostart.center)
                     -- ([#1]\tikztotarget.center) \tikztonodes
                     -- (\tikztotarget)},
}}

\newcommand{\pullbackangle}[2][]{\arrow[phantom,to path={
                     -- ($ (\tikztostart)!1cm!#2:([xshift=8cm]\tikztostart) $)
                        node[anchor=west,pos=0.0,rotate=#2,
                        inner xsep = 0]
                        {\begin{tikzpicture}[minimum
                        height=1mm,baseline=0,#1]
    \draw[-] (0,0) -- (.5em,.5em) -- (0,1em);
                        \end{tikzpicture}}}]{}}

\theoremstyle{definition}
\newtheorem{defn}[theorem]{Definition} %
\newtheorem{construction}[definition]{Construction}
\newtheorem{assumption}[remark]{Assumption}

\title{Coalgebra Encoding for Efficient Minimization}
\authorrunning{H.-P.~Deifel, S.~Milius, T.~Wi\ss\/mann}
\author{Hans-Peter Deifel}%
  {Friedrich-Alexander-Universität Erlangen-Nürnberg, Germany}%
  {hans-peter.deifel@fau.de}%
  {https://orcid.org/0000-0002-9542-9664}
  {Supported by the Deutsche Forschungsgemeinschaft (DFG) as part of the Research and Training Group 2475 ``Cybercrime and Forensic Computing'' (393541319/GRK2475/1-2019)}
\author{Stefan Milius}%
  {Friedrich-Alexander-Universität Erlangen-Nürnberg, Germany}%
  {stefan.milius@fau.de}%
  {https://orcid.org/0000-0002-2021-1644}
 {Supported by Deutsche Forschungsgemeinschaft (DFG) under project MI~717/5-2}%
\author{Thorsten Wi\ss\/mann}%
  {Radboud University Nijmegen, Netherlands}%
  {uni@thorsten.wissmann.de}%
  {https://orcid.org/0000-0001-8993-6486}{%
  Supported by NWO TOP project 612.001.852}%

\Copyright{Hans-Peter Deifel, Stefan Milius, Thorsten Wi\ss\/mann}


\hideLIPIcs %

\begin{document}
\FXRegisterAuthor{sm}{asm}{SM}%
\FXRegisterAuthor{hp}{ahp}{HP}%
\FXRegisterAuthor{tw}{atw}{TW}%

\maketitle
\begin{abstract}
  Recently, we have developed an efficient generic partition refinement
  algorithm, which computes behavioural equivalence on a state-based system
  given as an encoded coalgebra, and implemented it in the tool \copar{}. Here
  we extend this to a fully fledged minimization algorithm and tool by
  integrating two new aspects: (1)~the computation of the transition structure
  on the minimized state set, and (2)~the computation of the reachable part of
  the given system. In our generic coalgebraic setting these two aspects turn
  out to be surprisingly non-trivial requiring us to extend the previous
  theory. In particular, we identify a sufficient condition on encodings of
  coalgebras, and we show how to augment the existing interface, which encapsulates
  computations that are specific for the coalgebraic type functor, to make the
  above extensions possible. Both extensions have linear run time.
\end{abstract}

\section{Introduction}
\smnote{}
\label{sec:intro}
The task of minimizing a given finite state-based system has arisen in different
contexts throughout computer science and for various types of systems,
such as standard deterministic automata, tree
automata, transition systems,
Markov chains, probabilistic or other weighted
systems.
In addition to the obvious goal of reducing the mere memory consumption of the
state space, minimization often appears as a subtask of a more complex
problem. For instance, probabilistic model checkers benefit from minimizing the
input system before performing the actual model checking algorithm, as
e.g.~demonstrated in benchmarking by Katoen et al.~\cite{KatoenEA07}.

Another example is the graph isomorphism problem. A considerable portion of
input instances can already be decided correctly by performing a step called
colour refinement~\cite{BerkholzBG17}, which amounts to the minimization of a
weighted transition system wrt.~weighted bisimilarity.  \smnote{}

Minimization algorithms typically perform two steps: first a reachable subset of
the state set of the given system is computed by a standard graph search, and
second, in the resulting reachable system all behaviourally equivalent states
are identified. For the latter step one uses \emph{partition refinement} or
\emph{lumping} algorithms that start by identifying all states and then
iteratively refine the resulting partition of the state set by looking one step
into the transition structure of the given system. There has been a lot of
research on efficient partition refinement procedures, and the most efficient
algorithms for various concrete system types have a run time in
$\mathcal{O}(m\log n)$, for a system with $n$ states and $m$ transitions,
e.g.~Hopcroft's algorithm for deterministic automata~\cite{Hopcroft71} and the
algorithm by Paige and Tarjan~\cite{PaigeTarjan87} for transition systems, even
if the number of action labels is not fixed~\cite{Valmari09}. Partition
refinement of probabilistic systems also underwent a dynamic
development~\cite{CattaniS02,ZhangEA08}, and the best algorithms for Markov
chain lumping now match the complexity of the relational Paige-Tarjan
algorithm~\cite{HuynhTian92,DerisaviEA03,ValmariF10}.  For the minimization of
more complex system types such as Segala systems~\cite{BaierEM00,GrooteEA18}
(combining probabilities and non-determinism) or weighted tree
automata~\cite{HoegbergEA09}, partition refinement algorithms with a similar
quasilinear run time have been designed over the years.

Recently, we have developed a generic partition refinement
algorithm~\cite{dmsw17,concurSpecialIssue}\twnote{} and implemented it in the tool
\copar{}~\cite{coparFMSpecial,copartool}. This generic algorithm computes the
partition of the state set modulo behavioural equivalence for a wide variety of
stated-based system types, including all the above. This
genericity in the system type is achieved by working with \emph{coalgebras} for
a functor which encapsulates the specific types of transitions of the input
system. More precisely, the algorithm takes as input a syntactic description of
a set functor and an \emph{encoding} of a coalgebra for that functor and then
computes the simple quotient, i.e.~the quotient of the state set modulo
behavioural equivalence. The algorithm works correctly for every zippable set
functor (\autoref{def:zip}). It matches, and in some cases even improves on, the
run-time complexity of the best known partition refinement algorithms for many
concrete system types~\cite[Table~1]{coparFMSpecial}.

The reasons why this run-time complexity can be stated and proven generically
are: first, the encoding allows us to talk about the number of states and, in
particular, the number of transitions of an input coalgebra. But more
importantly, every iterative step of partition refinement requires only very
few system-type specific computations. These computations are encapsulated
in the \emph{refinement interface}~\cite{concurSpecialIssue}, which is then used
by the generic algorithm.

An important feature of our coalgebraic algorithm is its modularity: in the tool
the user can freely combine functors with already implemented refinement
interfaces by products, coproducts and functor composition. A refinement
interface for the combined functor is then automatically derived. In this way
more structured systems types such as (simple and general) Segala systems and
weighted tree automata can be handled.

In the present paper, we extend our algorithm to a fully fledged minimizer. In
previous work~\cite{AdamekEA14} it has been shown that for set functors
preserving intersections, every coalgebra equipped with a point, modelling
initial states, has a minimization called the \emph{well-pointed
  modification}. Well-pointedness means that the coalgebra does not have any
proper quotients (i.e.~it is \emph{simple}) nor proper pointed subcoalgebras
(i.e.~it is \emph{reachable}), in analogy to minimal deterministic automata
being reachable and observable~(see e.g.~\cite[p.~256]{ArbibManes86}).  The
well-pointed modification is obtained by taking the reachable part of the simple
quotient of a given pointed coalgebra~\cite{AdamekEA14} (and the more usual
reversed order, simple quotient of the reachable part, is correct for 
functors preserving inverse images~\cite[Sec.~7.2]{Wissmann2020}). Our previous
work on coalgebraic minimization algorithms has focused on computing the simple
quotient. Here we extend our algorithm by two missing aspects of minimization
and provide their correctness proofs: the computation of~(1) the transition
structure of the minimized system, and~(2) the reachable states of an input
coalgebra.

One may wonder why~(1) is a step worth mentioning at all because for many
concrete system types this is trivial, e.g.~for deterministic automata where the
transitions between equivalence classes are simply defined by choosing
representatives and copying their transitions from the input automaton.
However, for other system types this step is not that obvious, e.g.~for weighted
automata where transition weights need to be summed up and transitions might
actually disappear in the minimized system because weights cancel out. We found
that in the generic coalgebraic setting enabling the computation of the
(encoding of) the transition structure of the minimized coalgebra is
surprisingly non-trivial, requiring us to extend the theory behind our
algorithm.

\renewcommand{\propositionautorefname}{Prop.}%
\renewcommand{\theoremautorefname}{Thm.}%
In order to be able to perform
this computation generically we work with \emph{uniform encodings}, which are
encodings that satisfy a coherence property (\autoref{def:uniform-encoding}). We
prove that all encodings used in our previous work are uniform, and that the
constructions enabling modularity of our algorithm preserve uniformity
(\autoref{prop:prod-coprod-enc-uniform}). We also prove that uniform encodings
are subnatural transformations, but the converse does not hold in general. In
addition, we introduce the \emph{minimization interface} containing the new
function $\merge$ (to be implemented together with the refinement interface for
each new system type) which takes care of transitions that change as a result of
minimization. We provide $\merge$ operations for all functors with explicitly
implemented refinement interfaces (\autoref{ex:mininter}), and show that for
combined system types minimization interfaces can be automatically derived
(\autoref{prop:merge-for-prod}); similarly as for refinement interfaces. Our
main result is that the (encoded) transition structure of the minimized
coalgebra can be correctly computed in linear time
(\autoref{constr:correct:linear}).
\renewcommand{\propositionautorefname}{Proposition}
\renewcommand{\theoremautorefname}{Theorem}

Concerning extension~(2), the computation of reachable states, it is well-known
that every pointed coalgebra has a reachable part (being the smallest
subcoalgebra)~\cite{AdamekEA14,WissmannEA19}. Moreover, for a set functor
preserving intersections it coincides with the reachable part of the canonical
graph of the coalgebra~\cite[Lem.~3.16]{AdamekEA14}. Recently, it was shown that
the reachable part of a pointed coalgebra can be constructed
iteratively~\cite[Thm.~5.20]{WissmannEA19} and that this corresponds to
performing a standard breadth-first search on the canonical graph. The missing
ingredient to turn our previous partition refinement algorithm into a minimizer
is to relate the canonical graph with the encoding of the input coalgebra. We
prove that for a functor with a subnatural encoding, the encoding (considered as a graph)
of every coalgebra coincides with its canonical graph (\autoref{prop:her}).
\twnote{}

Putting everything together, we obtain an algorithm that computes the
well-pointed modification of a given pointed coalgebra. Both additions can be
implemented with linear run time in the size of the input coalgebra and hence do
not add to the run-time complexity of the previous partition refinement
algorithm. We have provided such an implementation with the new version of our tool
\copar{}.

All proofs and additional details can be found in the \fullorax{}.
\subparagraph*{Reachability in Coalgebraic Minimization}\hpnote{}
There are several works
on coalgebraic minimization, ranging from abstract constructions to concrete and
implemented algorithms~\cite{abhkms12,D,concurSpecialIssue,KoenigKuepper18,
  coparFMSpecial}, that compute the simple quotient~\cite{Gumm03} of a given
coalgebra. These are not concerned with reachability since coalgebras are not
equipped with initial states in general.

In Brzozowski's automata minimization algorithm~\cite{Brzozowski62},
reachability is one of the main ingredients. This is due to the duality of
reachability and observability described by Arbib and Manes~\cite{ArbibManes75},
and this duality is used twice in the algorithm. Consequently, reachability also
appears as a subtask in the categorical generalizations of Brzozowski's
algorithm~\cite{BonchiEA12,D,E,F,G}. These generalizations concern
automata processing input words and so do not cover minimization of
(weighted) tree automata. Segala systems are not treated either.
Due to the dualization, Brzozowki's classical algorithm for deterministic
automata has doubly exponential time complexity in the worst case (although it
performs well on certain types of non-deterministic
automata, compared to determinization followed by
minimization~\cite{TabakovVardi05}).%
\section{Background}%
\label{sec:prelim}
Our algorithmic framework~\cite{concurSpecialIssue} is defined on the level of
coalgebras for set functors, following the paradigm of \emph{universal
  coalgebra}~\cite{rutten_universal_2000}. Coalgebras can model a wide variety
of %
systems.

In the following we recall standard notation for sets and functions as well as
basic notions from the theory of coalgebras. We fix a singleton set $1=\{*\}$;
for each set $X$, we have a unique map $!\colon X\to 1$. We denote the disjoint
union (coproduct) of sets~$A,B$ by $A+B$ and use $\inl,\inr$ for the canonical
injections into the coproduct, as well as $\pr_1,\pr_2$ for the projections out
of the product. We use the notation $\fpair{\cdots}$, respectively~$[\,\cdots]$, for
the unique map induced by the universal property of a product, respectively coproduct.
We also fix two sets $2=\{0, 1\}$ and $3=\{0,1,2\}$ and use the former as a set
of boolean values with $0$ and $1$ denoting \emph{false} and \emph{true},
respectively. For each subset $S$ of a set $X$, the \emph{characteristic
  function} $\chi_S\colon X\to 2$ assigns 1 to elements of $S$ and 0 to elements
of $X\setminus S$. We denote by $\Set$ the category of all sets and maps. We
shall indicate injective and surjective maps by $\monoto$ and~$\epito$,
respectively.%

Recall that an endofunctor $F\colon \Set\to\Set$ assigns to each set~$X$ a
set~$FX$, and to each map $f\colon X\to Y$ a map $Ff\colon FX\to FY$, preserving
identities and composition, that is we have ~$F\id_X=\id_{FX}$ and
$F(g\cdot f)=Fg\cdot Ff$. We denote the composition of maps by $\cdot$ written
infix, as usual. An \emph{$F$-coalgebra} is a pair $(X,c)$ that consists of a
set~$X$ of \emph{states} and a map $c\colon X\to FX$ called \emph{(transition)
  structure}. A \emph{morphism} $h\colon (X,c)\to (Y,d)$ of $F$-coalgebras is a
map $h\colon X\to Y$ preserving the transition structure,
i.e.~$Fh\cdot c = d\cdot h$. Two states $x,y\in X$ of a coalgebra $(X,c)$ are
\emph{behaviourally equivalent} if there exists a coalgebra morphism $h$ with
$h(x) = h(y)$.

\begin{example}\label{ex:functors} Coalgebras and the generic notion for
  behavioural equivalence instantiate to a variety of well-known system types
  and their equivalences:
  \begin{enumerate}
  \item The \emph{finite powerset functor} $\Powf$ maps a set to the set of all
    its finite subsets and functions $f\colon X\to Y$ to
    $\Powf f = f[-]\colon\Powf X\to \Powf Y$ taking direct images. Its coalgebras are finitely
    branching (unlabelled) transition systems and coalgebraic behavioural
    equivalence coincides with Milner and Park's (strong) bisimilarity.
    
  \item\label{ex:fun:monoid} Given a commutative monoid $(M, +, 0)$, the \emph{monoid-valued functor}
    $M^{(-)}$ maps a set $X$ to the set of finitely supported functions from $X$
    to $M$. These are the maps $f\colon X\to M$, such that $f(x) = 0$ for all
    except finitely many $x\in X$. Given a map $h\colon X\to Y$ and a finitely
    supported function $f\colon X\to M$, $M^{(h)}(f)\colon M^{(X)}\to M^{(Y)}$
    is defined as $M^{(h)}(f)(y)=\sum_{x\in X, h(x)=y}f(x)$. Coalgebras for
    $M^{(-)}$ correspond to finitely branching weighted transition systems with
    weights from $M$. If a coalgebra morphism $h\colon (X,c)\to (Y,d)$ merges
    two states $s_1,s_2$, then for all transitions $x\xrightarrow{m_1} s_1$,
    $x\xrightarrow{m_2} s_2$ in $(X,c)$ there must be a transition
    $h(x)\xrightarrow{m_1+m_2} h(s_1) = h(s_2)$ in $(Y,d)$ and similarly if more
    than two states are merged. Coalgebraic behavioural equivalence captures weighted
    bisimilarity~\cite[Prop.~2]{Klin09}.

    Note that the monoid may have inverses: if $s_2 = -s_1$, then the
    transitions in the above example cancel each other out, leading to a
    transition $h(x)\xrightarrow{0}h(s_1)$ with weight 0, which in fact
    represents the absence of a transition. This happens for example for the
    monoid $(\R, +, 0)$ of real numbers. A simple minimization algorithm for
    real weighted transition (i.e.~$\R^{(-)}$-coalgebras) systems is given
    by Valmari and Franceschinis~\cite{ValmariF10}. These systems subsume Markov
    chains which are precisely the coalgebras for the finite probability
    distribution functor $\Dist$, a subfunctor of $\R^{(-)}$.

  \item Given a signature $\Sigma$ consisting of operation symbols $\sigma$,
    each with a prescribed natural number, its \emph{arity} $\ar(\sigma)$, the
    \emph{polynomial functor} $F_\Sigma$ sends each set $X$ to the set of
    (shallow) terms over $X$, specifically to the set
    \[
    \{\sigma(x_1, \ldots, x_n)\mid \sigma\in\Sigma, \ar(\sigma)=n, (x_1,\ldots,
    x_n)\in X^n\}.
    \]
    The action of $F$ on a function $f\colon X\to Y$ is given by
    \[
    F_\Sigma f(\sigma(x_1,\ldots,x_n)) = \sigma(f(x_1),\ldots,f(x_n)).
    \]
    A coalgebra structure $c\colon X\to F_\Sigma X$ assigns to a state $x \in X$
    an expression $\sigma(x_1,\ldots,x_n)$, where $\sigma$ is an output symbol
    and $x_1$ to $x_n$ are the successor states. Two states are behaviourally
    equivalent if their tree-unfoldings, obtained by repeatedly applying the
    coalgebra structure $c$, yields the same (infinite) $\Sigma$-tree.
    
  \item For a fixed alphabet $A$, the functor given by $FX = 2\times X^A$ is a
    special case of a polynomial functor over a signature with two symbols of
    arity $|A|$. An $F$-coalgebra $c\colon X \to 2 \times X^A$ is the same as a
    deterministic automaton without an initial state: the structure $c$ assigns
    a pair $(b,t)$ to each $x\in X$, where the boolean value $b \in 2$
    determines its finality, and the function $t\colon A \to X$ assigns to each
    input letter from $a \in A$ the successor state of $x$ under $a$. Here,
    behavioural equivalence coincides with language equivalence in the usual
    automata theoretic sense.
    
  \item The \emph{bag functor} $\Bag$ sends a set $X$ to the set of finite
    multisets over $X$ and functions $f\colon X\to Y$ to $\Bag f\colon \Bag X\to
    \Bag Y$ given by
    $\Bag f(\mbraces{x_1,\ldots,x_2}) = \mbraces{f(x_1),\ldots,f(x_2)}$,
    where we use the multiset braces $\lmbrace$ and $\rmbrace$ to differentiate
    from standard set notation; in particular $\mbraces{x,x}\neq
    \mbraces{x}$. Coalgebras for $\Bag$ are finitely branching transition
    systems where multiple transitions between any two states are allowed, or
    equivalently, weighted transition systems with positive integers as
    weights. This follows from the fact that the bag functor is (naturally
    isomorphic to) the monoid-valued functor for the monoid $(\N, +,
    0)$. Hence, behavioural equivalence coincides with weighted bisimilarity
    again.

    Note that every undirected graph may be considered as a $\Bag$-coalgebra by
    turning every edge into two directed edges with weight 1. Then two states
    are behaviourally equivalent iff they are identified by
    \emph{colour refinement}, also called the 1-dimensional Weisfeiler-Lehman
    algorithm (see e.g.~\cite{BerkholzBG17,CaiEA1992,Weisfeiler1976}).

  \end{enumerate}
\end{example}

\begin{example}[Modularity] \label{ex:modularity} New system types can be
  constructed from existing ones by functor composition. For example, labelled
  transition systems (LTSs) are coalgebras for the functor
  $FX = \Powf(A\times X)$, which is the composite of $\Powf$ and $A\times -$ for
  a label alphabet $A$, and precisely the bisimilar states in an $F$-coalgebra
  are behaviourally equivalent. Composing further, \emph{Segala systems} (or
  \emph{probabilistic LTSs}~\cite{GrooteEA18}) are coalgebras for
  $FX = \Powf(A\times\Dist X)$, for which coalgebraic behavioural equivalence
  instantiates to probabilistic bisimilarity~\cite{BartelsEA04}. Another example
  are \emph{weighted tree automata}~\cite{HoegbergEA09} with weights in a
  commutative monoid~$M$ and input signature $\Sigma$; they are coalgebras for
  the composed functor $FX=M^{(\Sigma X)}$, for which behavioural equivalence
  coincides with \emph{backwards bisimilarity}~\cite{coparFM19}.
\end{example}

\subparagraph*
{Simple, Reachable, and Well-Pointed Coalgebras}

Minimizing a given pointed coalgebra means to
compute its well-pointed modification. We now briefly recall the corresponding
coalgebraic concepts. For a more detailed and well-motivated discussion with
examples, see e.g.~\cite[Sec.~9]{AdamekEA20}.

First, a \emph{quotient coalgebra} of an $F$-coalgebra $(X,c)$ is represented by
a surjective $F$-coalgebra morphism, for which we write
$q\colon (X,c) \epito (Y,d)$, and a \emph{subcoalgebra} of $(X,c)$ is
represented by an injective $F$-coalgebra
morphism~$m\colon (S,s) \monoto (X,c)$.

A coalgebra $(X,c)$ is called \emph{simple} if it does not have any proper
quotient coalgebras~\cite{Gumm03}. That is, every quotient $q\colon (X,c) \epito (Y,d)$ is an
isomorphism. Equivalently, distinct states~$x,y \in X$ are never
behaviourally equivalent. Every coalgebra has an (up to isomorphism) unique
simple quotient (see e.g.~\cite[Prop.~9.1.5]{AdamekEA20}).
\begin{example}
  \begin{enumerate}
  \item\sloppypar A deterministic automaton regarded as a coalgebra for
    $FX = 2 \times X^A$ is simple iff it is
    observable~\cite[p.~256]{ArbibManes86}, that is, no distinct states accept
    the same formal language.
    
  \item A finitely branching transition system considered as a
    $\Powf$-coalgebra is simple, if it has no pairs of strongly bisimilar
    but distinct states; in other words if two states $x, y$ are strongly
    bisimilar, then $x = y$.
    
  \item A similar characterization holds for monoid-valued functors (such as the
    bag functor) wrt.~weighted bisimilarity.
  \end{enumerate}
\end{example}  

\noindent A \emph{pointed} coalgebra is a coalgebra $(X,c)$ equipped with a point
$i\colon 1 \to X$, equivalently a distinguished element $i\in X$, modelling an
initial state. Morphisms of pointed coalgebras are the point-preserving
coalgebra morphisms, i.e.~morphisms $h\colon (X,c,i) \to (Y,d,j)$
satisfying $h \cdot i = j$. Quotients and subcoalgebras of pointed coalgebras
are defined wrt.~these morphisms. A pointed coalgebra $(X,c,i)$ is called
\emph{reachable} if it has no proper subcoalgebra, that is, every subcoalgebra
$m\colon (S,s,j) \monoto (X,c,i)$ is an isomorphism. Every pointed coalgebra has a
unique reachable subcoalgebra (see e.g.~\cite[Prop.~9.2.6]{AdamekEA20}). The
notion of reachable coalgebras corresponds well with graph theoretic
reachability in concrete examples. We elaborate on this a bit more in
\autoref{sec:reach}.
\begin{example}
  \begin{enumerate}
  \item A deterministic automaton considered as a pointed coalgebra for $FX = 2 \times
    X^A$ (with the point given by the initial state) is reachable if all of
    its states are reachable from the initial state.
  \item A pointed $\Powf$-coalgebra is a finitely branching directed graph with a root
    node. It is reachable precisely when every node is reachable from the root
    node.
  \item Similarly, for monoid-valued functors such as the bag functor,
    reachability is precisely graph theoretic reachability, where a
    transition weight of 0 means `no edge'.
  \end{enumerate}
\end{example}

\noindent Finally, a pointed coalgebra $(X,c,i)$ is \emph{well-pointed} if it is reachable
and simple. Every pointed coalgebra has a \emph{well-pointed modification},
which is obtained by taking the reachable part of its simple quotient
(see~\cite[Not.~9.3.4]{AdamekEA20}).
\begin{remark}
  For a functor preserving inverse images,\smnote{}  one may
  reverse the two constructions: the well-pointed modification is the simple
  quotient of the reachable part of a given pointed
  coalgebra~\cite[Sec.~7.2]{Wissmann2020}. This is the usual order in which
  minimization of systems is performed algorithmically. However, for a functor
  that does not preserve inverse images, quotients of reachable coalgebras need
  not be reachable again~\cite[Ex.~5.3.27]{Wissmann2020}, possibly rendering the
  usual order incorrect.\smnote{}
\end{remark}
\smnote{}%
Our present paper is concerned with the \emph{minimization problem} for
coalgebras, i.e.~the problem to compute the well-pointed modification of a
given pointed coalgebra in terms of its encoding.

\begin{remark}\label{rem:reduction}
  Recall that a (sub)natural transformation $\sigma$ from a functor $F$ to a
  functor~$G$ is a set-indexed family of maps $\sigma_X\colon FX \to GX$ such
  that for every (injective) function $m\colon X \to Y$ the square on the right
  below commutes; we also say that $\sigma$ is \emph{(sub)natural in $X$}.
  
  From previous results (see \cite[Prop.~2.13]{concurSpecialIssue}
  and~\cite[Thm.~4.6]{WissmannEA19}) one obtains the following sufficient
  condition for reductions of reachability and simplicity. Given a family of
  maps $\sigma_X\colon FX\to GX$, then every $F$-coalgebra $(X,c)$ yields a
  $G$-coalgebra $(X, \sigma_X \cdot c)$ and we can reduce minimization tasks
  from $F$-coalgebras to $G$-coalgebras as follows:

  \vspace{2pt plus 1pt minus 1pt}
  \noindent
  \begin{minipage}[b]{.73\textwidth}%
  \begin{enumerate}
  \item\label{rem:reduction:1} Suppose that $\sigma\colon F \to G$ is
    \emph{sub-cartesian}, that is the squares on the right are
    pullbacks for every injective map $m\colon X \monoto Y$.  Then the reachable
    part of a pointed $F$-coalgebra $(X,c,i)$ is obtained from the reachable
    part of the $G$-coalgebra $(X, \sigma_X\cdot c, i)$.
  \end{enumerate}
  \end{minipage}%
  \hfill%
  \begin{minipage}[b]{.25\textwidth}
    \hspace*{0pt}\hfill
      \begin{tikzcd}[baseline=(FX.base)]
        |[alias=FX]|
        FX
        \arrow{r}{\sigma_X}
        \arrow{d}[swap]{Fm}
        &
        GX
        \arrow{d}{Gm}
        \\
        FY
        \arrow{r}{\sigma_Y}
        &
        GY
      \end{tikzcd}
      \\
      ~
  \end{minipage}
  \begin{enumerate}
    \stepcounter{enumi}
  \item Suppose that $F$ is a subfunctor of $G$, i.e.~we have a natural
    transformation $\sigma$ with injective components $\sigma_X\colon FX \monoto
    GX$. Then the problem of computing the simple
    quotient for~$F$-coalgebras reduces to that for $G$-coalgebras: the simple
    quotient of $(X, \sigma_X \cdot c)$ yields that of $(X,c)$.
  \end{enumerate}
  Consequently, if $F$ is a subfunctor of $G$ via a subcartesian $\sigma$, the
  minimization problem for $F$-coalgebras reduces to that for $G$-coalgebras.
  For example, the distribution functor $\Dist$ is a subcartesian subfunctor
  of $\R^{(-)}$. (For details see the \fullorax{}.)
\end{remark}

\subparagraph*
{Preliminaries on Bags}

The bag functor defined in \autoref{ex:functors} plays an important role in our
minimization algorithm, not only as one of many possible system types, but bags
are also used as a data structure. To this end, we use a couple of additional properties of
this functor.

\begin{remark}\label{rem:bags}
  \begin{enumerate}
  \item  Since $\Bag$ can also be regarded
    as a monoid-valued functor for $(\N, +, 0)$, every bag
    $b = \mbraces{x_1,\ldots,x_n}\in\Bag X$ may be identified with a finitely
    supported function $X \to \N$, assigning to each $x\in X$ its multiplicity
    in $b$. We shall often make use of this fact and represent bags as functions.

  \item The set $\Bag X$ itself is a commutative monoid with bag-union as
    the operation and the empty bag $\emptybag$ as the identity element.  In fact,
    this is the free commutative monoid over $X$. It therefore makes sense to
    consider the monoid-valued functor $(\Bag X)^{(-)}$ for a monoid of
    bags. Note that for every pair of sets $A, X$, the set $(\Bag A)^{(X)}$
    of finitely supported functions from $X$ to $\Bag A$ is isomorphic to
    $\Bag(A\times X)$ as witnessed by the following isomorphism (where
    $\swap$, $\curry$ and $\uncurry$ are the evident canonical
    bijections):
    \begin{align*}
      \group &= \big(\Bag(A \times X) \xrightarrow{\Bag(\swap)} \Bag(X
      \times A) \xrightarrow{\curry} (\Bag A)^{(X)}\big),\ \text{and}
      \\
      \ungroup &= \big((\Bag A)^{(X)} \xrightarrow{\uncurry} \Bag(X \times
      A) \xrightarrow{\Bag(\swap)} \Bag(A \times X)\big).
    \end{align*}
    Note that since $\swap$ is self-inverse and $\curry$, $\uncurry$ are
    mutually inverse, $\group$ and $\ungroup$ are mutually inverse, too. In
    symbols: 
    \begin{equation}
    \label{eq:group-ungroup}
    \group\cdot\ungroup = \id_{(\Bag A)^{(X)}},\quad
    \ungroup\cdot\group = \id_{\Bag(A\times X)}.
  \end{equation}
\end{enumerate}
\end{remark}

\noindent We often need to filter a bag of tuples $\Bag(A\times X)$ by a subset
$S\subseteq X$. To this end we define the maps
$\fil_S\colon\Bag(A\times X)\to \Bag(A)$ for sets $S\subseteq X$ and $A$ by
\[
  \fil_S(f)= \big(a\mapsto \sum_{x\in S}f(a, x)\big)
  =\mbraces{a \mid (a,x) \in f, x\in S},
\]
where the multiset comprehension is given for intuition.

\subparagraph*
{Zippable Functors}

One crucial ingredient for the efficiency of the generic
partition refinement algorithm~\cite{concurSpecialIssue} is that the coalgebraic
type functor is zippable:
\begin{notheorembrackets}
  \begin{defn}[{\cite[Def.~5.1]{concurSpecialIssue}}]\label{def:zip}%
    A set functor $F$ is called
    \emph{zippable} if the following maps are injective for every pair $A, B$ of
    sets:
  \[
    F(A+ B)
    \xrightarrow{~\fpair{F(A + !), F(! + B)}~}
    F(A + 1) \times F(1 + B).
  \]
\end{defn}
\end{notheorembrackets}
\smnote{}%
Zippability of a functor allows that partitions are refined incrementally by the
algorithm~\cite[Prop.~5.18]{concurSpecialIssue}, which in turn is the key for
allowing a low run time complexity of the implementation. For additional visual
explanations of zippability, see~\cite[Fig.~2]{concurSpecialIssue}.  We shall
need this notion in the proof of \autoref{encoding-injective}, and later proofs
use this result.

It was shown in~\cite{concurSpecialIssue} that all functors in \autoref{ex:functors} are
zippable. In addition, zippable functors are closed under products, coproducts
and subfunctors. However, they are not closed under functor composition,
e.g.~$\Powf\Powf$ is not zippable~\cite[Ex.~5.10]{concurSpecialIssue}.

\subparagraph*%
{The Trnkov\'a Hull}\label{trnkova}
For purposes of universal coalgebra, we may assume without loss of generality
that set functors preserve injections. Indeed, every set functor preserves
nonempty injections (being the split monomorphisms in $\Set$). As shown by
Trnkov\'a~\cite[Prop.~II.4 and~III.5]{Trnkova71}, for every set functor $F$
there exists an essentially unique set functor $\bar F$ which coincides with $F$
on nonempty sets and functions, and preserves finite intersections (whence
injections).  The functor $\bar F$ is called the \emph{Trnkov\'a hull} of
$F$. Since $F$ and $\bar F$ coincide on nonempty sets and maps, the categories
of coalgebras for $F$ and $\bar F$ are isomorphic.

\section{Coalgebra Encodings}
\label{sec:refint}
In order to make abstract coalgebras tractable for computers and to have a
notion of the size of a coalgebra structure in terms of nodes and edges as for
standard transition systems, our algorithmic framework encodes coalgebras using
a graph-like data structure. To this end, we require functors to be
equipped with an encoding as follows.
\begin{notheorembrackets}
\begin{defn} \label{def-encoding}
  An \emph{encoding} of a set functor $F$ consists of a set $A$ of \emph{labels} and a
  family of maps $\flat_X\colon FX\to\Bag(A\times X)$, one for every set $X$,
  such that the following map is injective:
  \[FX\xra{\fpair{F!, \flat_X}} F1\times\Bag(A\times X).\]
  \noindent
  An \emph{encoding} of a coalgebra $c\colon X\to FX$ is given by
  $\fpair{F!, \flat_X}\cdot c\colon X\to F1\times \Bag(A\times X)$.
\end{defn}
\end{notheorembrackets}
Intuitively, the encoding $\flat_X$ of a functor $F$ specifies how an
$F$-coalgebra should be represented as a directed graph, and the required
injectivity models that different coalgebras have different representations.
\begin{remark}
  Previously~\cite[Def.~6.1]{concurSpecialIssue}, the map $\fpair{F!, \flat_X}$
  was not explicitly required to be injective. Instead, a family of maps
  $\flat_X\colon FX\to \Bag(A\times X)$ and a \emph{refinement interface} for
  $F$ was assumed. The definition of a refinement interface for $F$ is tailored
  towards the computation of behaviourally equivalent states and its details are
  therefore not relevant for the present work. All we need here is that the
  existence of a refinement interface implies the injectivity condition of
  \autoref{def-encoding} and consequently, we inherit all examples of encodings
  from the previous work:
\end{remark}
\begin{proposition}\label{encoding-injective}%
  For every zippable set functor $F$ with a family of maps $\flat_X\colon FX\to
  \Bag(A\times X)$ and a refinement interface, the family $\flat_X$ is an
  encoding for $F$.
\end{proposition}

\begin{example}\label{ex:encodings}
  We recall a number of encodings from~\cite{concurSpecialIssue}; the
  injectivity is clear, and in fact implied by \autoref{encoding-injective}:
  \begin{enumerate}
  \item\label{ex:encodings:1} Our encoding for the finite powerset functor
    $\Powf$ resembles unlabelled transition systems by taking the singleton set
    $A=1$ as labels.  The map
    $\flat_X\colon\Powf(X)\to\Bag(1\times X)\cong\Bag(X)$ is the obvious
    inclusion, i.e.~$\flat_X(t)(*, x) = 1$ if $x\in t$ and 0 otherwise.
  \item\label{ex:encodings:2} The monoid-valued functor $M^{(-)}$ has labels
    from $A=M$ and
    $\flat_X\colon M^{(X)}\to\Bag(M\times X)$ is given by $\flat_X(t)(m, x) = 1$
    if $t(x) = m\neq 0$ and 0 otherwise.
  \item\label{ex:encodings:3} For a polynomial functor $F_\Sigma$, we use $A=\N$ as the label set and
    define the maps $\flat_X\colon F_\Sigma X\to\Bag(\N\times X)$ by
    $\flat_X(\sigma(x_1,\ldots,x_n)) = \mbraces{(1,x_1),\ldots,(n,x_n)}$.
    
    Note that $\flat_X$ itself is not injective if $\Sigma$ has at least two
    operation symbols with the same arity. E.g.~for DFAs ($F_\Sigma X = 2\times
    X^A$), $\flat_X$ only retrieves information about successor states but
    disregards the `finality' of states. However, pairing $\flat_X$ with
    $F!\colon FX\to F1$ yields an injective map.

  \item\label{ex:encodings:4} The bag functor $\Bag$ itself also has $A=\N$ as
    labels and $\flat_X(t)(n, x) = 1$ if $t(x) = n$ and 0 otherwise. This is
    just the special case of the encoding for a monoid-valued functor
    for the monoid $(\N, +, 0)$.
  \end{enumerate}
\end{example}

\noindent The encoding does by no means imply a reduction of the problem of minimizing
$F$-coalgebras to that of coalgebras for $\Bag(A\times -)$
(cf.~\autoref{rem:reduction}). In fact, the notions of behavioural equivalence
for $F$-coalgebras and coalgebras for $\Bag(A\times -)$, can be
radically different. If $\flat_X$ is natural %
in $X$, then behavioural
equivalence wrt.~$F$ implies that for $\Bag(A \times -)$, but not necessarily
conversely. However, we do not assume naturality of $\flat_X$, and in fact it fails in
all of our examples except one:
\begin{proposition}\label{prop:poly-flat-nat}
  The encoding $\flat_X\colon F_\Sigma X\to\Bag(A\times X)$ for the
  polynomial functor $F_\Sigma$ is a natural transformation.
\end{proposition}
\begin{example}\label{ex:nonnat}
  The encoding $\flat_X\colon \Powf(X) \to \Bag(1 \times X)\cong \Bag(X)$ in
  \autoref{ex:encodings} item \ref{ex:encodings:1} is not natural. Indeed, consider
  the map $!\colon 2 \to 1$, for which we have
  \[
    \Bag(!) \cdot \flat_2(\{0,1\})
    =
    \Bag(!)\mbraces{0,1}
    =
    \mbraces{*,*}
    \neq
    \mbraces{*}
    =
    \flat_1(\{*\})
    =
    \flat_1\cdot \Powf (!)(\{0,1\}).
  \]
  Similar examples show that the encodings in \autoref{ex:encodings}
  item~\ref{ex:encodings:2} (for all non-trivial monoids)
  and item~\ref{ex:encodings:4} are not natural.
\end{example}

An important feature of our algorithm and tool is that all implemented functors
can be combined by products, coproducts and functor composition. That is, the
functors from \autoref{ex:encodings} are implemented directly, but the algorithm also
automatically handles coalgebras for more complicated combined functors, like
those in \autoref{ex:modularity}, e.g.~$\Powf(A\times -)$. The mechanism that
underpins this feature is detailed in previous work
\cite{concurSpecialIssue,coparFM19} and depends crucially on the ability to form
coproducts and products of encodings:

\begin{notheorembrackets}
\begin{construction}[{\cite{concurSpecialIssue,coparFM19}}]\label{const:comb-enc}
  Given a family of functors $(F_i)_{i\in I}$ with encodings
  $(\flat_{X,i})_{i\in I}$ and $(A_i)_{i\in I}$, we obtain the following
  encodings with labels $A = \coprod_{i\in I} A_i$:
  \begin{enumerate}
  \item for the coproduct functor $F = \coprod_{i\in I} F_i$ we take
    \[
      \flat_X\colon \coprod_{i\in I} F_i X
      \xrightarrow{\coprod_{i\in I}\flat_{X,i}}\coprod_{i\in I}\Bag(A_i\times X)
      \xrightarrow{[\Bag(\inj_i\times X)]_{i\in I}}\Bag\big(\coprod_{i\in I} A_i\times X\big).
    \]
    
  \item  for the product functor $F = \prod_{i\in I} F_i$ we take
    \begin{align*}
      &\flat_X\colon \prod_{i\in I} F_i X \to \Bag(\coprod_{i\in I} A_i\times X)
      &\flat_X(t)(\inj_i(a), x) = \flat_i(\pr_i(t))(a, x),
    \end{align*}
    where $\inj_i\colon A_i \to \coprod_j A_j$ and $\pr_i\colon \prod_j F_jX \to
    F_iX$ denote the canonical coproduct injections and product projections, respectively.
  \end{enumerate}
\end{construction}
\end{notheorembrackets}
\begin{proposition}\label{P:prod-coprod-enc}
  The families $\flat_X$ defined in \autoref{const:comb-enc} yield encodings for
  the functors $\prodi F_i$ and $\coprodi F_i$, respectively.
\end{proposition}
\begin{remark}\label{rem:comp}
  Since zippable functors are not closed under composition, modularity cannot be
  achieved by simply providing a construction of an encoding for a composed
  functor (at least not without giving up on the efficient run-time
  complexity). Functor composition is reduced to coproducts making a detour via
  many-sorted sets. Here is a rough explanation of how this works. Suppose that
  $F$ is a \emph{finitary} set functor, which means that for every $x \in FX$
  there exists a finite subset $Y \subseteq X$ and $x' \in FY$ such that
  $x = Fm (x')$ for the inclusion map $m\colon Y \subto X$. Given a finite coalgebra
  $c\colon X \to FG X$, it can be turned into a $2$-sorted coalgebra
  $(c',d')\colon (X, Y) \to (FY, GX)$ as follows: since $F$ is finitary one
  picks a finite subset~$Y$ of $GX$ such that there exists a map $c'\colon X \to
  FY$ with $c= Fd' \cdot c'$, where $d'\colon Y \subto GX$ is the inclusion
  map. Then $c'$ and $d'$ are combined into one coalgebra on the disjoint union
  $X+Y$ as shown below:
  \[
    X + Y
    \xrightarrow{~c'+d'~}
    FY + GX
    \xrightarrow{~[F\inr,G\inl]~}
    (F+G)(X+Y)
  \]
  for the coproduct of the functors $F$ and $G$, where $\inl\colon X \to X + Y$ and
  $\inr\colon Y \to X+Y$ are the two coproduct injections. Full
  details may be found in~\cite[Sec.~8]{concurSpecialIssue}.
\end{remark}

For the sake of computing the coalgebra structure of the minimized coalgebra,
we require that, intuitively,
the labels used for encoding $FX$ are independent of the cardinality of $X$:
\begin{defn}\label{def:uniform-encoding}
  An encoding $\flat_X$ for a set functor $F$ is called
  \emph{uniform} if it fulfils the following property for every~$x\in X$:
  \begin{equation}
    \label{eq:uniform-encoding}
    \begin{tikzcd}[row sep = 5]
      FX \ar[r, "\flat_X"]\ar[dd, swap, "F\chi_{\{x\}}"] & \Bag(A\times
      X)\ar[dr, "\fil_{\{x\}}", near start]\\
      & & \Bag(A)\\
      F2 \ar[r, "\flat_2"] & \Bag(A\times 2)\ar[ur, "\fil_{\{1\}}", swap, near start]
    \end{tikzcd}
  \end{equation}
\end{defn}

Intuitively, the condition in \autoref{def:uniform-encoding} expresses that in
an encoded coalgebra, the edges (and their labels) to a state $x$
do not change if other states $y,z \in X\setminus\{x\}$ are identified by a
possible partition on the state space. Diagram \eqref{eq:uniform-encoding}
expresses the extreme case of such a partition, particularly the one where
\emph{all} elements of $X$ except for $x$ are identified in a block, with
$x$ being in a separate singleton block.

Fortunately, requiring uniformity does not exclude any of the existing encodings
that we recalled above.
\begin{proposition}\label{prop:enc}
  All encodings from \autoref{ex:encodings} are uniform.
\end{proposition}
Uniform encodings interact nicely with the modularity constructions:
\begin{proposition}\label{prop:prod-coprod-enc-uniform}
  Uniform encodings are closed under product and coproduct.
\end{proposition}%
That is, given functors $(F_i)_{i\in I}$ with uniform encodings
$(\flat_i)_{i\in I}$, then the encodings for the functors $\coprod_{i\in I} F_i$
and $\prod_{i\in I} F_i$, as defined in \autoref{const:comb-enc}, are uniform.

Admittedly, the condition in \autoref{def:uniform-encoding} is slightly
technical. However, we will now prove that it sits strictly between two standard
properties, naturality and subnaturality.
\begin{proposition}\label{prop:subnat}
  \begin{enumerate}
  \item\label{prop:subnat:1} Every natural encoding is uniform.
  \item\label{prop:subnat:2} Every uniform encoding is a subnatural transformation.
  \end{enumerate}
\end{proposition}
The converses of both of the above implications fail in general. For the
converse of~\ref{prop:subnat:1} we saw a counterexample in \autoref{ex:nonnat},
and for the converse of~\ref{prop:subnat:2} we have the following
counterexample. 
\begin{example}\label{ex:nonsub}
  Consider the following encoding for the functor $FX=X\times X\times X$
  given by $A = 3 + 3$ and
  \begin{align*}
    &\flat_X\colon FX\to\Bag(A\times X)\\
    &\flat_X(x, y, z) =
    \begin{cases}
      \{(\inl 0, x), (\inl 1, y), (\inl 2, z)\} & \text{if }y = z,\\
      \{(\inr 0, x), (\inr 1, y), (\inr 2, z)\} & \text{if }y \not= z.\\
    \end{cases}
  \end{align*}
  This encoding is subnatural, since the value of $y=z$ is preserved by
  injections under $F$. But it is not uniform, for if $x\neq y\neq z$, then we have
  \begin{equation*}
    \fils{1}(\flat(F\chi_{\{x\}}(x, y, z))) = \fils{1}(\flat(1, 0, 0)) = \{\inl 0\} \neq \{\inr 0\} = \fils{x}(\flat(x, y, z)).
  \end{equation*}
\end{example}

\section{Computing the Simple Quotient}
\label{sec:simpl}

The previous coalgebraic partition refinement algorithm and its tool
implementation in \copar{} compute for a given encoding of a coalgebra $(X,c)$
the state set of its simple quotient $q\colon (X,c) \epito (Y, d)$, that is the
partition $Y$ of the set $X$ corresponding to behavioural equivalence. But the
algorithm does not compute the coalgebra structure $d$ of the simple quotient
(and note that it is not given the structure $c$ explicitly, to begin
with). Here we will fill this gap. We are interested in computing the
encoding $Y \xrightarrow{d} FY \xrightarrow{\flat_Y}\Bag(A\times Y)$ given the
encoding $X\xrightarrow{c} FX \xrightarrow{\flat_X}\Bag(A\times X)$ of the input
coalgebra and the quotient map $q\colon X \epito Y$.

The edge labels in the encoding of the quotient coalgebra relate to the labels
in the encoded input coalgebra in a functor specific way. For example, for
weighted transition systems, the labels are the transition weights, which are
added whenever states are identified. In contrast, for deterministic automata
(or when $F$ is a polynomial functor), the labels (i.e.~input symbols) on the
transitions remain the same even when states are identified.

Thus, when computing the encoding of the simple quotient, the modification of
edge labels is functor specific. Algorithmically, this is reflected by
specifying a new interface containing one function $\merge$, which is intended
to be implemented together with the refinement interface~(\autoref{sec:refint})
for every functor of interest. The abstract function $\merge$ is then used in
the generic \autoref{cons:main} in order to compute the encoding of the simple
quotient.

\begin{defn}
  A \emph{minimization interface} for a set functor $F$ equipped with a functor encoding $\flat_X: FX\to\Bag(A\times X)$ is a function
  $\merge\colon \Bag(A) \to \Bag(A)$
  such that the following diagram commutes for all $S\subseteq X$:
  \begin{equation}
    \label{eq:merge-axiom}
    \begin{tikzcd}[column sep=12mm]
      FX \ar[r, "\flat_X"]
      \ar[d, "F\chi_S",swap]
      &
      \Bag(A \times X)
      \ar[r, "\fil_S"]
      &
      \Bag(A)\ar[d, "\merge"]
      \\
      F2 \ar[r, "\flat_2"]
      &
      \Bag(A\times 2)
      \ar[r, "\fil_{\{1\}}"]
      & \Bag(A)
    \end{tikzcd}
  \end{equation}
\end{defn}

Intuitively, $\merge$ expresses what happens on the labels of edges from one
state to one block. It receives the bag of all labels of edges from a particular
source state $x$ to a \emph{set of} states $S$ that the minimization procedure
identified as equivalent. It then computes the edge labels from $x$ to the
merged state $S$ of the minimized coalgebra in a functor specific
way. \autoref{fig:merge-intuition} depicts this process for a monoid-valued
functor (cf.~\autoref{ex:functors}, item~\ref{ex:fun:monoid}).
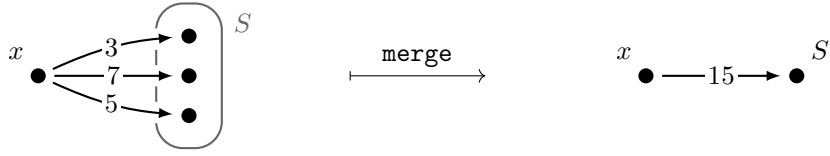
\begin{figure}[h]
  \centering
  \begin{tikzpicture}[%
    dot/.style={circle,fill,inner sep=.2em},%
    block/.style={rounded corners=10,color=black!60},
    lab/.style={fill=white,inner sep=1pt},
    ->,>=latex,thick,
    transition/.style={
      shorten <= 1mm,
      shorten >= 1mm,
      bend angle=10,
      preaction={
        draw=white, %
        -, %
        line width=8pt,
        shorten <=4mm,
      },
    },
    ]
    \begin{scope}
      \node[dot] (x) [label=120:$x$] {};
      \node[dot] (s2) [right=5em of x] {};
      \node[dot] (s1) [above=3mm of s2] {};
      \node[dot] (s3) [below=3mm of s2] {};
      \draw[block] ($(s1.north west)+(-1em,1em)$) rectangle ($(s3.south east)+(1em,-1em)$);
      \path (x) edge[transition,bend left] node[lab] {3} (s1);
      \path (x) edge[transition] node[lab] {7} (s2);
      \path (x) edge[transition,bend right] node[lab] {5} (s3);
      \node [above right of=s2,color=black!60] {$S$};
    \end{scope}
    \node (merge) [right=2cm of s2,minimum width=5em] {};
    \draw[commutative diagrams/.cd, every arrow,mapsto]
      (merge.west) -- node [above] {$\merge$} (merge.east);
    \begin{scope}
      \node[dot] (xneu) [right=2cm of merge, label=120:$x$] {};
      \node[dot] (sneu) [right=5em of xneu, label=60:$S$] {};
      \draw (xneu) edge[transition] node[lab] {15} (sneu);
    \end{scope}
  \end{tikzpicture}
  \caption{Example application of $\merge$ for the monoid-valued functor.}
  \label{fig:merge-intuition}
\end{figure}
In this example, $\merge$ sums up the labels (which are monoid elements),
resulting in a correct transition label to the new merged state.

Before we give formal definitions of $\merge$ for the functors of interest, let us show that
there is a close connection between properties of $\merge$ and the encoding;
this will simplify the definition of $\merge$ later (\autoref{ex:mininter}).

First, if $\merge$ receives the bag of labels from a source state to a
\emph{single} target state, then there is nothing to be merged and thus $\merge$
should simply return its input bag. Moreover, we can even characterize uniform
encodings by this property:
\begin{lemma}\label{L:merge-enc}
  Given a minimization interface, the following are equivalent:
  \begin{enumerate}
  \item\label{L:merge-enc:1} $\merge(\fil_{\{x\}}(\flat_X(t))) = \fil_{\{x\}}(\flat_X(t))$ for all $t\in FX$.
  \item $\flat_X$ is uniform.
  \end{enumerate}
\end{lemma}
\smallskip
\noindent
Similarly, the property that $\merge$ is \emph{always} the identity
characterizes \emph{natural} encodings:
\begin{lemma}\label{lem:merge-id-enc-nat}
  For every encoding $\flat_X\colon FX\to\Bag(A\times X)$, the following are equivalent:
  \begin{enumerate}
  \item\label{item:merge-id-enc-nat} The identity on $\Bag A$ is a minimization interface.
  \item\label{item:enc-nat-merge-id} $\flat_X$ is a natural transformation.
  \end{enumerate}
\end{lemma}

\begin{example}\label{ex:mininter}
  \begin{enumerate}
  \item For the finite powerset functor $\Powf(-)$, with labels $A = 1$, we define
    $\merge{}\colon \Bag 1\to \Bag 1$ by
    $\merge(\ell)(*) = \min(1, \ell(*))$.
  \item For monoid-valued functors $M^{(-)}$ with $A = M$, \merge{} is defined as
    \[
      \merge(\ell) =
      \begin{cases}
        \mbraces{ \Sigma\ell } & \Sigma \ell\neq 0\\
        \emptybag & \text{otherwise,}
      \end{cases}
    \]
    where $\Sigma\colon \Bag(M) \to M$ is defined
    by $\Sigma\mbraces{m_1, \ldots, m_n} = m_1 + \cdots + m_n$. 
  \item The encoding for the polynomial functor $F_\Sigma$ for a signature
    $\Sigma$ is a natural transformation and hence its minimization interface
    is given by $\merge=\id$ (see \autoref{lem:merge-id-enc-nat}).
  \end{enumerate}
\end{example}
\begin{proposition}\label{prop:all-merges-go}
  All $\merge$ maps in \autoref{ex:mininter} are minimization interfaces and run
  in linear time in the size of their input bag.
\end{proposition}

\noindent
Having $\merge$ defined for the functors of interest, we can now use it to
compute the encoding of the simple quotient.

\begin{assumption}\label{ass:simpl}
  For the remainder of this section we assume that $F1\neq \emptyset$.
\end{assumption}
This is w.l.o.g.\ since $F1 = \emptyset$ if and only if $FX=\emptyset$ for all
sets $X$, for which there is only one coalgebra (which is therefore its own
simple quotient already).

\begin{proposition}\label{thm:merge-thm}
  Suppose that the set functor $F$ is equipped with a uniform
  encoding $\flat_X\colon FX\to\Bag(A\times X)$ and a minimization interface
  \merge{}. Then the diagram below commutes for every map $q\colon X\to Y$,
  \begin{equation}
    \label{eq:merge-thm}
    \begin{tikzcd}[column sep = 35]
      FX \ar[r, "\flat_X"]\ar[d, "Fq",swap]
        & \Bag(A\times X)\ar[r, "\Bag(A\times q)"]
        & \Bag(A\times Y)\ar[r, "\group"] \ar[d, dashed]
        & \Bag(A)^{(Y)} \ar[d, "\merge^{(Y)}"]\\
      FY \ar[rr, "\flat_Y"]
        &
        & \Bag(A\times Y)
        & \Bag(A)^{(Y)} \ar[l, "\ungroup",swap]
    \end{tikzcd}
  \end{equation}
\end{proposition}
Note that the dashed arrow is not simply the identity map because~$\flat_X$
fails to be natural for most functors of interest (\autoref{ex:nonnat}).
\begin{proof}[Proof (Sketch)]
  One first proves that $\merge$ preserves empty bags:
  $\merge(\mbraces{}) = \mbraces{}$. The commutativity of the desired
  diagram~\eqref{eq:merge-thm} is proven by extending it by every evaluation map
  $\ev(y)\colon \Bag(A)^{(Y)} \to \Bag(A)$, $y \in Y$, which form a jointly
  injective family. The extended diagram for $y \in Y$ is then proven
  commutative using~\eqref{eq:uniform-encoding} for $y$, \eqref{eq:merge-axiom}
  for $S = q^{-1}[y]$, which is also used in the form $\chi_{\{y\}} \cdot q
  = \chi_S$ in addition to two easy properties of $\ev$ and
  $\fil$: $\fil_{\{y\}} = \ev(y)\cdot \group$ and $\fil_{\{y\}} \cdot \Bag(A
  \times q) = \fil_S$.
\end{proof}

\begin{construction}\label{cons:main}
  \sloppypar Given the encoded $F$-coalgebra $(X,\,\flat_X\!\cdot\! c)$, the
  quotient $q\colon X \epito Y$, and a minimization interface for $F$, we define the map
  $e\colon Y \to \Bag(A\times Y)$ as follows: given an element $y\in Y$, choose
  any $x \in X$ with $q(x) = y$ and put
  \[
    e(y) := (\ungroup\cdot\,\merge{{}^{(Y)}}\!\cdot \group \cdot\, \Bag(A\times
    q)\cdot \flat_X\cdot c)(x),
  \]
  where the involved types are as follows:
\begin{equation}\label{eq:goal}
\begin{tikzcd}
  X \ar[r, "c"]\ar[->>,d, swap, "q"] & FX \ar[r, "\flat_X"] &
  \Bag(A\times X)\ar[r, "\Bag(A\times q)"] &[3mm] \Bag(A\times Y)\ar[r, "\group"] & \Bag(A)^{(Y)} \ar[d, "\merge^{(Y)}"]\\
  Y \ar[rr, "e"] & & \Bag(A\times Y) & & \Bag(A)^{(Y)} \ar[ll, "\ungroup",swap]
\end{tikzcd}
\end{equation}
\end{construction}

For the well-definedness and the correctness of \autoref{cons:main}, we need to
prove that~\eqref{eq:goal} commutes. Moreover, observe that $c$ is not
directly given as input, and that the structure~$d\colon Y\to FY$ of the simple
quotient is not computed; only their encodings $\flat_X \cdot c$ and
$e = \flat_Y \cdot d$ are.

\begin{theorem} \label{constr:correct:linear}
  Suppose that $q\colon (X,c) \epito (Y,d)$ represents a quotient
  coalgebra. Then \autoref{cons:main} correctly yields the encoding $e = \flat_Y
  \cdot d$ given the encoding $\flat_X \cdot c$ and the partition of $X$
  associated to $q$.  

  If $\merge$ runs in linear time (in its parameter), then \autoref{cons:main}
  can be implemented with linear run time (in the size of the input coalgebra $\flat_X \cdot c$).
\end{theorem}
\noindent
\twnote{}%
In the run time analysis, a bit of care is needed so that the implementation of
$\group$ has linear run time; see the \fullorax{} for details.  From \autoref{prop:all-merges-go} we
see that for every functor from \autoref{ex:functors}, \autoref{cons:main} can
be implemented with linear run time.

\subsection{Modularity of Minimization Interfaces}

Modularity in the system type is gained by reducing functor
composition to products and coproducts (\autoref{rem:comp}). Since we want the
construction of the minimized coalgebra structure to benefit from the same
modularity, we need to verify closure under product and coproduct for the
notions required in \autoref{thm:merge-thm}. We have already done so
for uniform encodings (\autoref{prop:prod-coprod-enc-uniform}); hence it remains
to show that minimization interfaces can also be combined by product and coproduct:

\begin{construction}\label{C:merge-prod}
  Given a family of functors $(F_i)_{i\in I}$ together with uniform encodings
  $\flat_i\colon F_iX\to\Bag(A_i\times X)$ and minimization interfaces
  $\merge_i\colon \Bag(A_i)\to\Bag(A_i)$, we define \merge{} for the (co)product
  functors  $\prodi F_i$ and $\coprodi F_i$ as follows: 
  \[
    \begin{array}{l}
      \merge\colon\Bag(\coprodi A_i)\to\Bag(\coprodi A_i)
      \qquad
      \merge(t)(\inj_i a) = \merge_i(\filter_i(t))(a),
    \end{array}
  \]
  where $\filter_i\colon\Bag(\coprod_{j\in I}A_j)\to\Bag(A_i)$ is
  given by $\filter_i(f)(a) = f(\inj_i(a))$.
\end{construction}
Curiously, the definition of $\merge$ is the same for products and
coproducts, e.g.~because the label sets are the same
(see~\autoref{const:comb-enc}). However, the correctness proofs turns out to be
quite different. Note that for coproducts, all labels in the image of
$\fil_S\cdot\,\flat_X$ are in the same coproduct component. Thus, $\filter_i$
never removes elements and acts as a mere type-cast when the above $\merge$ is
used in accordance with its specification.
\begin{proposition}\label{prop:merge-for-prod}
  The \merge{} function defined in \autoref{C:merge-prod} yields a minimization
  interface for the functors $\prodi F_i$ and $\coprodi F_i$. It can be
  implemented with linear run-time if each $\merge_i$ is linear in its input.
\end{proposition}
\begin{corollary}
  The class of set functors having a minimization interface contains all
  polynomial and all  monoid-valued functors and is closed under
  product and coproduct. 
\end{corollary}
\noindent
Consequently, \autoref{cons:main} correctly yields encoded quotient
coalgebras for those functors. Note that all functors from
\autoref{ex:mininter} are contained in this class. Furthermore, functor
composition can be dealt with by using coproducts as explained in
\autoref{rem:comp}.

\section{Reachability}
\label{sec:reach}

Having quotiented an encoded coalgebra by behavioural equivalence, the remaining
task is to restrict the coalgebra to the states that are actually reachable from
a distinguished initial state. For an intersection preserving set
functor, the reachable part of a pointed coalgebra can be constructed
iteratively, and this reduces to standard graph search on the canonical graph of
the coalgebra~\cite[Cor.~5.26f]{WissmannEA19}, which we now recall.
Throughout, $\Pow$ denotes the (full) powerset functor. The following is 
inspired by Gumm~\cite[Def.~7.2]{Gumm2005}:
\begin{defn}%
  \label{def:cangr}
  Given a functor $F\colon\Set\to\Set$, we define a family of maps
  $\tau^F_X\colon FX\to \Pow X$ by
  $
    \tau_{X}^F(t) = \{ x\in X\mid 1\xrightarrow{t}FX\text{ does not
      factorize through }F(X\setminus\{x\})\xrightarrow{Fi}FX \},
  $
  where $i\colon X\setminus\{x\}\hookrightarrow X$ denotes the inclusion
  map.

  The \emph{canonical graph} of a coalgebra $c\colon X\to FX$ is the directed graph
  $
    X\xrightarrow{c} FX \xrightarrow{\tau_{X}^F} \Pow X.
  $
  The nodes are the states of $(X,c)$ and one has an edge from $x$ to $y$
  whenever $y \in \tau^F_X (c(x))$. 
\end{defn}
\noindent
Note that for a pointed coalgebra $(X,c,i)$ its canonical graph is equipped with
the same point $i\colon 1 \to X$, that is, the canonical graph is equipped with
a root node $i(*) \in X$. As we pointed out in \autoref{sec:prelim}, reachability
of the pointed $\Pow$-coalgebra $(X,\tau^F_X\cdot c, i)$ precisely means that
every $x \in X$ is reachable from the root node in the canonical graph.
\begin{example}\label{ex:cangr}
  \begin{enumerate}
  \item For a deterministic automaton considered as a coalgebra for $FX = 2
    \times X^A$ the canonical graph is precisely its usual underlying state
    transition graph. %

  \item For the finite powerset functor $\Powf$, it is easy to see that
    $\tau^{\Powf}_X\colon \Powf X \subto \Pow X$ is the inclusion map. Thus,
    the canonical graph of a $\Powf$-coalgebra (a finitely branching graph) is
    itself.

  \item\label{ex:cangr:3} For the functor $\Bag(A \times -)$ the maps
    $\tau^{\Bag(A\times -)}_X\colon \Bag(A \times X) \to \Pow X$ act as follows
    \[
      \mbraces{(a_1,x_1), \ldots, (a_n, x_n)}
      \mapsto
      \{x_1, \ldots,  x_n\}.
    \]
    Hence, if we view a coalgebra $X \to \Bag(A\times X)$ as a
    finitely-branching graph whose edges are labelled by pairs of elements of
    $A$ and $\N$, then the canonical graph is that same graph but without the
    edge labels. This holds similarly also for other monoid-valued functors.
  \end{enumerate}
\end{example}

\noindent To perform reachability analysis on encoded coalgebras, we would like that the
canonical graph of a coalgebra and its encoding coincide. This clearly follows when,
given a set functor~$F$ with encoding $\flat_X\colon FX\to\Bag(A\times X)$, the
following equation holds for every set $X$:\!\!%
\begin{equation}\label{eq:desire}
  \tau_X^F = \big(FX \xrightarrow{~\flat_X~} \Bag(A \times X)
  \xrightarrow{~\tau_X^{\Bag(A\times -)}~} \Pow X\big).
\end{equation}%

\begin{assumption} \label{ass:reachable}
  For the rest of this section we assume that $F$ is an intersection
  preserving set functor equipped with a subnatural encoding $\flat_X\colon FX
  \to \Bag(A \times X)$.
\end{assumption}
\begin{remark}
  That $F$ preserves intersections is an extremely mild condition for set
  functors. All the functors in \autoref{ex:encodings} preserve
  intersections. Furthermore, the collection of intersection preserving set
  functors is closed under products, coproducts, and functor composition. A
  subfunctor $\sigma\colon F \monoto G$ of an intersection preserving functor
  $G$ preserves intersections if $\sigma$ is a cartesian natural transformation,
  that is all naturality squares are pullbacks (cf.~\autoref{rem:reduction}).

  Let us note that for every finitary set functor (cf.~\autoref{rem:comp}) the Trnkov\'a
  hull $\bar F$ (see~p.~\pageref{trnkova}) preserves
  intersections~\cite[Cor.~8.1.17]{AdamekEA20}.
\end{remark}

We are now ready to show the desired equality~\eqref{eq:desire} by point-wise
inclusion in either direction. Under the running \autoref{ass:reachable} it
follows that the encoding of a coalgebra can only mention states that are in the
coalgebra's canonical graph:
\begin{proposition}\label{prop:hin}
  For every $t\in FX$ we have that
  $\tau_X^{\Bag(A\times -)}(\flat_X(t)) \subseteq \tau_X^F(t)$.
\end{proposition}
\begin{proof}[Proof (Sketch)]
  This is shown by contraposition. If $x$ is not in $\tau_X^F(t)$, then we know that
  the map $t\colon 1 \to FX$ factorizes through $F(X\setminus\{x\})\xrightarrow{Fi}FX$
  (cf.~\autoref{def:cangr}). Using the subnaturality square of $\flat$ for the
  map $i$ then yields $x \not\in \tau_X^{\Bag(A\times -)}(\flat_X(t))$.
\end{proof}
For the converse inclusion, we additionally require that $F$ meets the
assumptions of the partition refinement algorithm:

\begin{theorem}\label{prop:her}
  The
  canonical graph of a finite coalgebra coincides with that of its encoding. 
\end{theorem}
For every finite set $X$ one proves the equation~\eqref{eq:desire}: $\tau_X^F =
\tau_{X}^{\Bag(A\times -)} \cdot\,\flat_X$. It suffices to prove the reverse of
the inclusion in \autoref{prop:hin} -- again by contraposition. This time the
argument is more involved using that the map $\fpair{F!,\flat_X}$ is injective
(\autoref{def-encoding}), and that $F$ preserves intersections.
(For details see the \fullorax{}.)

As a consequence of \autoref{prop:her}, the states in the reachable part of a
pointed coalgebra $(X, c, i)$ are precisely the states reachable from the node
$i(*) \in X$ in the (underlying graph of the) encoding
$\flat_X \cdot c\colon X \to \Bag(A \times X)$,
cf.~\autoref{ex:cangr}\ref{ex:cangr:3}. Thus, given (the encoding of) a pointed
coalgebra $(X,c,i)$, its reachable part can be computed in linear time by a
standard breadth-first search on the encoding viewed as a graph (ignoring the labels).

This holds for all the functors in \autoref{ex:encodings} and
every functor obtained from them by forming products, coproducts and functor
composition.

\section{Conclusions and Future Work}

We have shown how to extend a generic coalgebraic partition refinement algorithm
to a fully fledged minimization algorithm. Conceptually, this is the step from
computing the simple quotient of a coalgebra to computing the well-pointed
modification of a pointed coalgebra. To achieve this, our extension includes two
new aspects: (1)~the computation of the transition structure of the simple
quotient given an encoding of the input coalgebra and the partition of its state
space modulo behavioural equivalence, and (2)~the computation of the encoding of
the reachable part from the encoding of a given pointed coalgebra. Both of these
new steps have also been implemented in the Coalgebraic Partition Refiner
\copar, together with a new pretty-printing module that prints out the resulting
encoded coalgebra in a functor-specific human-readable syntax.%

There are a number of questions for further work. This mainly concerns
broadening the scope of generic coalgebraic partition refinement algorithms.
First, we will further broaden the range of system types that our algorithm and
tool can accommodate, and provide support for base categories beside the
sets as studied in the present work,
e.g.~nominal sets, which underlie nominal
automata~\cite{BojanczykEA14,SchroderEA17}.

\smnote{}
Concerning genericity, there is an orthogonal approach by Ranzato and
Tapparo~\cite{RanzatoT08}, which is variable in the choice of the \emph{notion
  of process equivalence} -- however within the realm of standard labelled
transition systems (see also~\cite{GrooteEA17}). Similarly, Blom and
Orzan~\cite{BlomOrzan03,BlomOrzan05} use a technique called \emph{signature
  refinement}, which handles strong and branching bisimulation as well as Markov
chain lumping (see also \cite{DijkPol18}).

To overcome the bottleneck on memory consumption that is inherent in partition
refinement~\cite{ValmariF10,Valmari09}, symbolic and distributed methods have
been employed for many concrete system
types~\cite{BergaminiEA05,BlomOrzan03,BlomOrzan05,GaravelHermanns02,Wijs15,DijkPol18}.
We will explore in future work whether these methods, possibly generic in the
equivalence notion, can be extended to the coalgebraic generality.

\subparagraph*{Acknowledgement} We would like to thank the anonymous referees
for their comments, which helped us to improve the presentation.
\twnote{}

\label{maintextend}

\bibliographystyle{plainurl}%
\bibliography{refs}

\clearpage
\input{appendix.tex}

\end{document}

%% file: appendix.tex
\clearpage
\appendix

\section{Additional Notation in Omitted Proofs}
\label{app:notation}
Recall that the quotients of a set $X$, represented by surjective maps
$X \epito P$ are in one-to-one correspondence with partitions on $X$. More
generally, every map $f\colon X\to Y$ induces an equivalence relation
\[
  \ker(f) = \{ (x_1,x_2) \in X \times X \mid f(x_1) = f(x_2)\}
\]
called the \emph{kernel} of $f$. If $f\colon X\epito Y$ is surjective, then $\ker(f)$ is the
equivalence relation corresponding to the partition $Y$ on $X$.

\section{Omitted Proofs}
\subsection{Proofs for \autoref{sec:prelim}}

\subsection*{Details for \autoref{rem:reduction}}

Recall that a (sub)natural transformation
$\sigma$ from a functor $F$ to a functor $G$ is a set-indexed family of maps
$\sigma_X\colon FX \to GX$ such that for every (injective) function $f\colon X \to Y$ we
have
\[
  \begin{tikzcd}
    FX
    \arrow{r}{\sigma_X}
    \arrow{d}[swap]{Ff}
    &
    GX
    \arrow{d}{Gf}
    \\
    FY
    \arrow{r}{\sigma_Y}
    &
    GY
  \end{tikzcd}
\]
As usual, we shall also say that the family $\sigma$ is \emph{(sub)natural in $X$}.

A subnatural transformation $\sigma\colon F \to G$ is called \emph{subcartesian} if the
above ``naturality squares'' are pullbacks for every injective map $f$. 

Given a natural transformation $\sigma\colon F \to G$ every $F$-coalgebra
$(X,c)$ yields a $G$-coalgebra $(X, \sigma_X\cdot c)$.  

Recall that $F$ is a \emph{subfunctor} of $G$ if there is a
natural transformation $\sigma\colon F \to G$ all of whose components
$\sigma_X\colon FX \monoto GX$ are injective maps.

\sloppypar
In the following proposition point~\ref{prop:reduction:1} is
standard, for point~\ref{prop:reduction:2}
see~\cite[Prop.~2.13]{concurSpecialIssue}, and point~\ref{prop:reduction:3} can
be gleaned from~\cite[Thm.~4.6]{WissmannEA19}. We provide a full proof for the
convenience of the reader. 

\begin{proposition}\label{prop:reduction}
  Let $\sigma\colon F \to G$ be natural transformation.
  \begin{enumerate}
  \item\label{prop:reduction:1} Behavioural equivalence wrt.~$F$ implies that for $G$. 
  \item\label{prop:reduction:2} If $F$ is a subfunctor of $G$ via $\sigma$, then
    the problem of computing the simple quotient for $F$-coalgebras reduces to
    that for $G$-coalgebras.
  \item\label{prop:reduction:3} If $\sigma$ is subcartesian, then the problem of
    computing the reachable part for pointed $F$-coalgebras reduces to that for
    pointed $G$-coalgebras.
  \end{enumerate}
\end{proposition}
\noindent
Consequently, if $F$ is a subfunctor of $G$ via a subcartesian $\sigma$ the
minimization problem for $F$-coalgebras reduces to that for $G$-coalgebras.
\begin{proof}
  \begin{enumerate}
  \item This follows from the fact that for every morphism $h\colon (X,c) \to
    (Y,d)$ of $F$-coalgebras we have the following commutative diagram due to
    the naturality of $\sigma$:
    \begin{equation}\label{eq:pres}
      \begin{tikzcd}
        X
        \arrow{r}{c}
        \arrow{d}[swap]{h}
        &
        FX
        \arrow{r}{\sigma_X}
        \arrow{d}{Fh}
        &
        GX
        \arrow{d}{Gh}
        \\
        Y
        \arrow{r}{d}
        &
        FY
        \arrow{r}{\sigma_Y}
        &
        GY
      \end{tikzcd}
    \end{equation}
    This      actually      shows      that      the      object      assignment
    $(X,c)  \mapsto  (X,\sigma_X  \cdot  c)$  is a  functor  from  the  category
    $\Coalg  F$  of  all  $F$-coalgebras  to the  category  $\Coalg  G$  of  all
    $G$-coalgebras, which acts as the identity on morphisms.
    
  \item We first prove that the above functor $\Coalg F \to \Coalg G$ preserves
    and reflects quotient coalgebras if $\sigma_X$ is injective.

    For preservation, note that every quotient $q\colon (Y,d) \epito (X,c)$
    yields the quotient
    $q\colon (Y,\sigma_Y \cdot d) \epito (X, \sigma_X \cdot c)$ wrt.~the functor
    $G$, cf.~Diagram~\eqref{eq:pres}.

    For reflection, let $(X,c)$ be an $F$-coalgebra and let
    $q\colon (X, \sigma_X \cdot c) \epito (Y, d')$ be any quotient of
    $G$-coalgebras. Since $q$ is surjective and $\sigma_Y$ injective we obtain a
    unique coalgebra structure $d\colon Y \to FY$ such that $q$ is a morphism of
    $F$-coalgebras:
    \[
      \begin{tikzcd}
        X
        \arrow{r}{c}
        \arrow[->>]{d}[swap]{q}
        &
        FX
        \arrow{d}{Fq}
        \arrow[>->]{r}{\sigma_X}
        &
        GX
        \arrow{d}{Gq}
        \\
        Y
        \arrow[dashed]{r}{d}
        \arrow[shiftarr = {yshift=-15}]{rr}{d'}
        &
        FY
        \arrow[>->]{r}{\sigma_Y}
        &
        GY
      \end{tikzcd}
    \]
    
    The desired reduction is now obvious since the simple quotient of an
    $F$-coalgebra $(X,c)$ coincides with that of the $G$-coalgebra $(X,\sigma_X
    \cdot c)$.

  \item We first prove that the functor $\Coalg F \to \Coalg G$ induced by
    $\sigma$ preserves and reflects pointed subcoalgebras.

    Preservation is clear
    by using Diagram~\eqref{eq:pres} and the fact that a morphism $h\colon
    (X,c,i) \to (Y,d,j)$ of pointed coalgebras preserves the point: $h \cdot i =
    j$.

    For reflection, let $(X,c, i)$ be any $F$-coalgebra and let
    $m\colon (S,s', j) \monoto (X, \sigma_X \cdot c, i)$ be a subcoalgebra. Then
    from the fact that $\sigma$ is subcartesian we obtain a unique coalgebra
    structure $s\colon S \to FS$ such that $m\colon (S,s,j) \to (X,c,i)$ is a
    pointed subcoalgebra wrt.~$F$:
    \[
      \begin{tikzcd}
        S
        \arrow[dashed]{r}{s}
        \arrow[>->]{d}[swap]{m}
        \arrow[shiftarr = {yshift=15}]{rr}{s'}
        &
        FS
        \arrow{r}{\sigma_S}
        \arrow{d}[swap]{Fm}
        \pullbackangle{-45}
        &
        GS
        \arrow{d}{Gm}
        \\
        X
        \arrow{r}{c}
        &
        FX
        \arrow{r}{\sigma_X}
        &
        GX        
      \end{tikzcd}
    \]

    This implies that the reachable parts of $(X,c,i)$ wrt.~$F$ and
    $(X,\sigma_X\cdot c, i)$ wrt.~$G$ coincide, which clearly establishes the
    desired reduction.\qedhere
  \end{enumerate}
\end{proof}

\noindent For further use we collect a few properties of the filter function $\fil_S$.
\begin{lemma}\label{lem:fils-properties}
  \begin{enumerate}
  \item\label{prop:fils-natural} The maps
    $\fil_S\colon \Bag(A \times X) \to \Bag(A)$ are natural in $A$.
  \item\label{L:ev}
    For every $x\in X$, we have
    \begin{equation}\label{eq:filsingle}
      \fil_{\{x\}} = \ev(x)\cdot\group,
    \end{equation}
    where $\ev: X \to Y^X \to Y$ is the evaluation of the exponential $Y^X$ (in curried form).
  \item\label{item:fils-eq-group} The function
    $\fpair{\fils{x}}_{x\in X}\colon \Bag(A\times X)\to (\Bag A)^X$ has a
    codomain restriction to $(\Bag A)^{(X)}$, and this is equal to the function
    $\group$.
    
  \item\label{lem:filsingle2} For every function $f\colon X\to Y$ and $S\subseteq X$,
    \begin{equation}
      \label{eq:filsingle2}
      \fil_S = \fil_{f[S]} \cdot\,\Bag(A\times f).
    \end{equation}
  \end{enumerate}
\end{lemma}
\begin{proof}%
  \begin{enumerate}
  \item This was proved in previous work~\cite[Rem.~6.5]{concurSpecialIssue}.
    
  \item Given $x\in X$, $a\in A$ and $f\in \Bag(A\times X)$,
    \begin{align*}
      \fils{x}(f)(a) &= \sum_{y\in \{x\}}f(a, y) = f(a, x)
      = (\lambda b. f(b, x))(a)\\
      &= \ev(x)(\lambda y.\lambda b. f(b, y))(a)
      = \ev(x)(\group(f))(a).
    \end{align*}

  \item The first part of the statement is clear, and the second part follows
    from \autoref{L:ev}.

  \item Given $t\in\Bag(A\times X), f\colon X\to Y, S\subseteq X$ and $a\in A$:
    \begin{align*}
      \fil_{f[S]}&(\Bag(A\times f)(t))(a) = \sum_{y\in f[S]}\Bag(A\times f)(t)(a, y)
      =\sum_{y\in f[S]}\sum_{\substack{x\in X\\f(x)=y}}t(a, x)\\
      &=\sum_{\mathclap{\substack{x\inj X\\f(x)\in f[S]}}}t(a, x)
      =\sum_{x\in S}t(a, x) = \fil_S(t)(a).\qedhere
    \end{align*}
  \end{enumerate}
\end{proof}

\subsection{Proofs for \autoref{sec:refint}}
\subsection*{Proof of \autoref{encoding-injective}}
We recall the definition of a refinement interface of a functor $F$ here for the convenience of
the reader. For detailed treatment, please refer to our previous work~\cite{concurSpecialIssue,coparFMSpecial}.
\begin{defn}
 Given sets $S\subseteq C\subseteq X$, the map $\chi_S^C\colon X\to 3$ is
 defined by
 \[
   \chi_S^C(x) = \begin{cases}
     2 &\text{if }x \in S, \\
     1 &\text{if }x \in C\setminus S, \\
     0 &\text{if }x \in X\setminus C. \\
   \end{cases}
 \]
 Intuitively, this is a `three-valued characteristic function' and is equivalent
 to the map $\fpair{\chi_S,\chi_C}\colon X\to 2\times 2$ without the impossible
 case $(1,0)$ ($x\in S$, $x\not\in C$, $S\subseteq C$).
\end{defn}
\begin{notheorembrackets}
\begin{defn}[\cite{concurSpecialIssue}]\label{def:refinementInterface}
  Given a set $A$ and a family of maps $\flat_X\colon FX\to \Bag(A\times X)$
  (one for every set $X$), a \emph{refinement interface} for a functor~$F$ is
  formed by a set $W$ of \emph{weights} and functions
  \begin{align*}
    \begin{array}{l@{\qquad}c@{\quad}l}
      \op{init}\colon F1\times\Bag A\to W,&\op{update}\colon \Bag A \times W \to W\times F3\times W
    \end{array}
  \end{align*}
  such that there exists a family of \emph{weight maps}
  $w\colon \Pow X \to (FX \to W)$ such that for all $S\subseteq C\subseteq X$,
  the diagrams
  \begin{equation*}
    \begin{tikzcd}
      F1 \times \Bag A
      \arrow{r}{\op{init}}
      & W
      \\
      FX
      \arrow{u}[left, pos=0.4]{
        \fpair{
          F!,
          \op{fil}_X\cdot \flat}}
      \arrow{ur}[swap]{w(X)}
    \end{tikzcd}
    \hspace{1mm}
    \begin{tikzcd}[column sep=1.4cm]
      \Bag A \times W
      \arrow{r}{\op{update}}
      & W \times F3 \times W
    \\
    FX
    \arrow[
    ]{ur}[swap]{\fpair{w(S),F\chi_S^C, w(C\setminus S)}}
    \arrow{u}[left]{\fpair{\op{fil}_S\cdot \flat,w(C)}}
  \end{tikzcd}
\end{equation*}
commute.
\end{defn}
\end{notheorembrackets}

\begin{notheorembrackets}
\begin{remark}[{\cite{concurSpecialIssue}}]
  The crucial part of forming a refinement interface for a functor $F$ is
  finding an appropriate set $W$ and maps $w\colon \Pow X\to (FX\to W)$.
  However, only the maps $\op{init}$ and $\op{update}$ are implemented, whereas
  $w$ just ensures correctness (and is not thus implemented). In most instances,
  we have 
  \[
    W := F2
    \qquad\text{and}\qquad
    w(C) := F\chi_C\colon FX\to F2
    \qquad\text{for }C\subseteq X.
  \]
  The only exception is the monoid valued functor $FX= M^{(X)}$ for a
  non-cancellative monoid $(M,+,0)$~\cite{coparFMSpecial} (so in particular also
  for $FX=\Powf X$~\cite{concurSpecialIssue}), where we have:
  \begin{align*}
    &W := M\times \Bag(M\setminus \{0\})
    \qquad\text{and}\\
    &w(C)\colon M^{(X)}\to M\times \Bag(M\setminus \{0\})\\
    &w(C)(\mu) = \big(\sum_{x\in X\setminus C}\mu(x),~~ m\mapsto |\{x\in X\mid
    \mu(x) = m\}| \big)
  \end{align*}
  where we define the bag $\Bag(M\setminus\{0\})$ by a map
  $(M\setminus\{0\})\to \N$. For $FX=\Pow X$, we have $M = 2$
  and $\Bag(2\setminus\{0\}) = \Bag(\{1\}) \cong \N$, hence $w(C)\colon \Powf
  X\to 2\times \N$ sends a successor structure $t\in \Powf X$ to $w(C)(t)\in
  2\times \N$ which provides (a) the information whether $t$ contains a successor outside of $C$
  and (b) the number of successors in $C$. Keeping track of the number of
  successors in the blocks $C$ of the partition is one of the main ideas of the
  $\CO(m\log n)$
  algorithm by Paige and Tarjan~\cite{PaigeTarjan87}.

  The functions $\op{init}, \op{update}$, which are to be implemented for every
  functor $F$ of interest, incrementally compute these weights (in $W$) and the
  three valued characteristic function $F\chi_S^C\colon FX\to F3$.
\end{remark}
\end{notheorembrackets}

\begin{lemma}\label{encoding2anyF3}
  Given a functor $F$ with a refinement interface, we have for all sets
  $S\subseteq C \subseteq X$ a map $r_S^C\colon F1\times \Bag(A\times X) \to F3$
  with $r_S^C\cdot \fpair{F!, \flat_X} = F\chi_S^C$.
  \[
    \begin{tikzcd}
      FX
      \arrow{rr}{F\chi_S^C}
      \arrow{dr}[swap]{\fpair{F!,\flat_X}}
      & & F3
      \\
      & F1\times \Bag(A\times X)
      \arrow{ur}[swap]{r_S^C}
    \end{tikzcd}
  \]
\end{lemma}
\begin{proof}
  \tikzset{
    extreme L shape/.style={
        to path={
          (\tikztostart.west)
          -- ([#1]\tikztostart.west)
          |- ([yshift=-4mm]\tikztotarget.south) \tikztonodes
          -- (\tikztotarget.south)
        },
        rounded corners,
    },
    extreme L shape/.default={
      xshift=-6mm
    },
  }
  First, we define maps $v_X$ and $v_C$ by the commutativity of the
  left-hand parts of the diagrams below, respectively. We also observe
  that precomposing these maps with $\fpair{F!,\flat_X}$ yields $w(X)$
  and $w(C)$, respectively, using the axioms of the refinement
  interface (note that the left-hand triangle in the right-hand
  diagram commutes by the left-hand diagram):
  \[
    \begin{tikzcd}[sep = 12mm]
      F1\times \Bag(A\times X)
      \arrow[extreme L shape={xshift=-2mm}]{dr}[pos=0.3,left]{v_X}[pos=0.3,right]{:=}
      \arrow{d}[swap]{F1\times \Bag \pr_1}
      \descto{dr}{Axiom \\ \init}
      & FX
      \arrow{l}[swap]{\fpair{F!,\flat_X}}
      \arrow{d}{w(X)}
      \\
      F1\times \Bag A
      \arrow{r}[swap]{\init}
      & W
    \end{tikzcd}
    \quad
    \begin{tikzcd}[column sep = 12mm, row sep=18mm]
      F1\times \Bag(A\times X)
      \arrow[extreme L shape={xshift=-2mm}]{dr}[pos=0.25,left]{v_C}[pos=0.25,right]{:=}
      \arrow{d}[description,pos=0.25]{\fpair{v_X, \fil_C\cdot \pr_2}}
      \descto[pos=0.75]{dr}{Axiom \\ \update}
      & FX
      \arrow{l}[swap]{\fpair{F!,\flat_X}}
      \arrow{d}{w(C)}
      \arrow{dl}[sloped,above]{\fpair{w(X),~ \fil_C\cdot\, \flat_X}}
      \\
      W \times \Bag A
      \arrow{r}[swap]{\pr_1\cdot\, \update}
      & W
    \end{tikzcd}
  \]
  Now we can define $r_S^C$ by the commutativity of the left-hand part
  in the diagram below and show that is has the desired property:
  \[
    \begin{tikzcd}[column sep = 12mm, row sep=18mm, baseline=(F3.base)]
      F1\times \Bag(A\times X)
      \arrow[extreme L shape={xshift=-9mm}]{dr}[pos=0.30,left]{r_S^C}[pos=0.30,right]{:=}
      \arrow{d}[swap]{\fpair{v_C, \fil_S\cdot \pr_2}}
      \descto[pos=0.75]{dr}{Axiom \\ \update}
      \descto[pos=0.15]{dr}{Def.~$v_C$}
      & FX
      \arrow{l}[swap]{\fpair{F!,\flat_X}}
      \arrow{d}{F\chi_S^C}
      \arrow{dl}[sloped,above]{\fpair{w(C),~ \fil_S\cdot\, \flat_X}}
      \\
      W \times \Bag A
      \arrow{r}[swap]{\pr_2\cdot\, \update}
      & |[alias=F3]| F3
    \end{tikzcd}
    \hfill\raisebox{-12pt}{\qedhere}
  \]
\end{proof}

\begin{proof}[Proof of \autoref{encoding-injective}]
  Let $X= \{x_0,\ldots,x_{n-1}\}$. We define the following family of
  subsets of $X$:
  \[
    S_i = \{x_i\},
    \qquad
    C_i = \{x_{i},\ldots,x_{n-1}\}
    \qquad \text{for $0\le i < n$}.
  \]
  For every $k$, $0\le k \le n$, we define the map
  \[
    q_k := \fpair{\chi_{S_i}^{C_i}}_{0\le i < k}\colon X\longrightarrow \prod_{0\le i < k} 3
    ~\cong~ 3^k.
  \]
  The partitions corresponding to $\chi_{S_i}^{C_i}$ and $q_k$ are illustrated
  in \autoref{fig:partitionSequence}. Note that the partition for
  $q_{k+1}$ is formed by the (nonempty) intersections of the blocks
  from the partitions for $q_k$ and $\chi_{S_k}^{C_k}$.
  \begin{figure}
    \centering
    \begin{tikzpicture}[
      partition block/.style={
        draw=black!50,
        line width=1pt,
        inner sep=0pt,
        minimum width=1.7em,
        minimum height=1.7em,
        rounded corners=0.85em,
        text=red, %
      }
      ]
      \foreach \partitionType/\xshift in {chi/0cm, q/6cm} {
        \begin{scope}[xshift=\xshift, x=7mm]
        \foreach \k in {0,1,2,3} {
          \foreach \i/\xlabel in {0/0,1/1,2/2,3/3,5/{n-1}} {
            \node (\partitionType\k x\i) at (\i,-\k) {$x_{\xlabel}$};
          }
          \tikzset{shift the last dots in chi/.style={
            }}
          \ifthenelse{\equal{\partitionType}{chi}}{
            \ifthenelse{\k = 3}{
              \tikzset{shift the last dots in chi/.style={
                  xshift=2mm,
                }}
            }{}
          }{}
          \draw[draw=none] (\partitionType\k x3) --
          node[shift the last dots in chi] (\partitionType\k x4) {$\cdots$}
          (\partitionType\k x5);
        }
        \end{scope}
      }
      \foreach \k in {0,1,2,3} {
        \node[anchor=east] at ([xshift=-2mm]chi\k x0.west) {$X/\chi_{S_\k}^{C_\k}:$};
        \node[anchor=east] at ([xshift=-2mm]q\k x0.west) {$X/q_\k:$};
      }
      \foreach \k in {0,1,2,3}{
        \ifthenelse{\k > 0}{
          \pgfmathsetmacro{\kMinusOne}{int(\k-1)}
          \node[partition block, fit=(chi\k x0) (chi\k x\kMinusOne)] {};
        }{}
        \node[partition block, fit=(chi\k x\k)] {};
        \pgfmathsetmacro{\kPlusOne}{int(\k+1)}
        \node[partition block, fit=(chi\k x\kPlusOne) (chi\k x5)] {};
        \ifthenelse{\k > 0}{
          \pgfmathsetmacro{\kMinusOne}{int(\k-1)}
          \foreach \i in {0,...,\kMinusOne} {
            \node[partition block, fit=(q\k x\i)] {};
          }
        }{}
        \node[partition block, fit=(q\k x\k) (q\k x5)] {};
      }
    \end{tikzpicture}
    \caption{The partitions for $q_k$ in the proof of \autoref{encoding-injective}}
    \label{fig:partitionSequence}
  \end{figure}
  Clearly, the union of the equivalence relations $\ker (q_k)$ and $\ker
  (\chi_{S_k}^{C_k})$ is again an equivalence relation, for every $0\le k < n$.
  Thus, we can apply~\cite[Prop.~5.18]{concurSpecialIssue} to obtain
  \[
    \ker F\fpair{q_k,\chi_{S_k}^{C_k}}
    = \ker \fpair{Fq_k, F\chi_{S_k}^{C_k}}
    \qquad
    \text{for all }0\le k < n.
  \]
  Combining these $n$-many equalities, we obtain
  \begin{equation}\label{eq:ker}
    \ker F q_n
    = \ker \fpair{F\chi_{S_k}^{C_k}}_{0\le k < n}.
  \end{equation}
  Here, the map $q_n\colon X\to 3^n$ is injective, because for every $x_i\in X$ the element
  $q_n(x_i) \in 3^n$ is clearly the only element in the image of $q_n$ that has
  $2$ in the $i$th component:
  \[
    \pr_i(q_n(x_i)) = \chi_{S_i}^{C_i}(x_i) = \chi_{S_i}^{C_i}(x_i)
    \overset{\text{Def.~}S_i}{=} \chi_{\{x_i\}}^{C_i}(x_i)
    \overset{\text{Def.~}\chi}{=} ~ 2.
  \]
  Since $F$ preserves injective maps, $Fq_n\colon FX\to F3^n$ is
  injective, too. Thus, so is $\fpair{F\chi_{S_k}^{C_k}}_{0\le k < n}$
  by equation~\eqref{eq:ker} of kernels (note that a map
  $f\colon A\to B$ is injective iff the relation $\ker(f)$ is the
  identity relation on $A$).

  Finally, we pair all the maps $r_{S_k}^{C_k}$, $0\le k < n$, that we
  have derived from the refinement interface in \autoref{encoding2anyF3}.
  We obtain the following commutative diagram:
  \[
    \begin{tikzcd}
      FX
      \arrow[>->]{rr}{\fpair{F\chi_{S_k}^{C_k}}_{0\le k < n}}
      \arrow{dr}[swap]{\fpair{F!,\flat_X}}
      & & \prod_{0\le k < n}F3
      \\
      & F1\times \Bag(A\times X)
      \arrow{ur}[swap]{\fpair{r_{S_k}^{C_k}}_{0\le k < n}}
    \end{tikzcd}
  \]
  We know that the map at the top is injective. It follows from the standard
  laws for injective maps that $\fpair{F!,\flat_X}$ is injective, as desired.
\end{proof}

\subsection*{Proof of \autoref{prop:poly-flat-nat}} 
\begin{proof}
  Given $f\colon X\to Y$ and $\sigma(x_1,\ldots,x_n)\in F_\Sigma X$, we calculate
  \begin{align*}
    \Bag(A\times f)(\flat_X(\sigma(x_1,\ldots,x_n)))
                   &= \Bag(A\times f)\{(1,x_1),\ldots,(n,x_n)\}\\
                   &= \{(1,f(x_1)),\ldots,(n,f(x_n))\}\\
                   &= \flat_Y(\sigma(f(x_1),\ldots,f(x_n)))\\
                   &= \flat_Y(F_\Sigma f(\sigma(x_1,\ldots,x_n))).\qedhere
  \end{align*}
\end{proof}
\subsection*{Proof of \autoref{P:prod-coprod-enc}}
\begin{proof} By the assumption on the individual functors $F_i$, we know that
  $\langle F_i!, \flat_{X,i}\rangle$ is injective for every $i\in I$. We now
  prove the required injectivity for the coproduct and product functors
  respectively:
  \begin{enumerate}
  \item For the coproduct $\coprodi F_i$, assume
    $\inj_jt_1, \inj_k t_2\in \coprodi F_iX$ such that
    $\langle F!,\flat_X\rangle(\inj_j t_1) = \langle F!,\flat_X \rangle(\inj_k
    t_2)$. Since this implies that $\inj_j (F_j!\; t_1) = \inj_k (F_k!\; t_2)$,
    we know that $j=k$. We now expand the definition
    of $\flat_X$ and calculate for $i = 1,2$:
    \begin{align*}
      \langle F!, \flat_X\rangle(\inj_j t_i)
      &= \langle F!, [\Bag(\inj_k\times \id)]_{k\in I}\cdot\coprodi[k]\flat_{X,k}\rangle(\inj_j t_i) \\
      &= \id\times[\Bag(\inj_k\times \id)]_{k\in I}\cdot\langle F!, \coprodi[k]\flat_{X,k}\rangle(\inj_j t_i) \\ 
      &= \inj_j\times\Bag(\inj_j\times \id)\cdot\langle F_j!, \flat_{X,j}\rangle(t_i). 
    \end{align*}
    Since $\inj_j$ is an injective map, $\Bag$ preserves injections, and
    injections are stable under product, we see
    that $\inj_j \times \Bag(\inj_j \times \id)$ is injective, whence so is its
    composite with $\fpair{F_j!, \flat_{X,j}}$. Since this composite merges
    $t_1$ and $t_2$ by assumption, we conclude $t_1 = t_2$.
    
  \item For the product $\prodi F_i$, assume $t_1,t_2\in \prodi F_iX$ such that
    $\langle F!,\flat_X\rangle(t_1) = \langle F!,\flat_X\rangle(t_2)$. We will
    show that $t_1=t_2$. The assumption implies that $F!(t_1)=F!(t_2)$ as well as
    $\flat_X(t_1)=\flat_X(t_2)$. The former implies that for all $i\in I$ we
    have $F_i!(\pr_i t_1) = F_i!(\pr_i t_2)$. From the latter and the definition of $\flat_X$ we
    obtain for every $i\in I$, $a \in A_i$ and $x \in X$ that
    \[
      \flat_{X,i}(\pr_i(t_1))(a,x)
      =
      \flat_X(t_1)(\inj_i a, x)
      =
      \flat_X(t_2)(\inj_i a, x)
      =
      \flat_{X,i}(\pr_i(t_2))(a,x).
    \]
    Consequently, for all $i\in I$ we have
    $\langle F_i!, \flat_{X,i} \rangle(\pr_i(t_1)) = \langle F_i!,\flat_{X,i}
    \rangle(\pr_i(t_2))$ which implies $\pr_i(t_1) = \pr_i(t_2)$ by the
    assumption. Hence, $t_1 = t_2$ since the projections $\pr_i$, $i \in I$, form a jointly
    monic family.\qedhere
  \end{enumerate}
\end{proof}

\subsection*{Proof of \autoref{prop:enc}}
\begin{proof}%
  First, note that for every bag $b\in \Bag(A \times X)$ and every pair $(a,x)
  \in A \times X$  we have $\fil_{\{x\}}(b)(a) = b(a, x)$.
  \begin{enumerate}
  \item For the finite powerset functor $\Powf(-)$, we have $A = 1$ and
    $\flat\colon\Powf X\to \Bag(1\times X)\cong \Bag(X)$ given by

    \[ \flat(t)(*,x) = \begin{cases}
        1 & x\in t,\\
        0 & \text{otherwise.}
      \end{cases}
    \]
    Observe that $1\in\Powf\chi_{\{x\}}(t) \Leftrightarrow x\in t$. We then have
    \begin{align*}
      \fil_{\{x\}}&(\flat(t))(*) = \flat(t)(*,x) = \begin{cases}
        1 & x\in t\\
        0 & \text{otherwise,}
      \end{cases}\\
      \fil_{\{1\}}&(\flat(\Powf\chi_{\{x\}}(t)))(*) = \flat(\Powf\chi_{\{x\}}(t))(*,1) = \begin{cases}
        1 & 1\in\Powf\chi_{\{x\}}(t)\\
        0 & \text{otherwise}
      \end{cases}\\
      &= \begin{cases}
        1 & x\in t\\
        0 & \text{otherwise.}
        \end{cases}
    \end{align*}
  \item The monoid-valued functor $M^{(-)}$ for a given monoid $M$ has labels $A= M$ and $\flat\colon M^{(X)}\to\Bag(M\times X)$ given by
    \[ \flat(t)(m,x) =
      \begin{cases}
        1 & t(x) = m\neq 0,\\
        0 & \text{otherwise.}
      \end{cases}
    \]
    Observe that $M^{(\chi_{\{x\}})}(t)(1) = t(x)$. We then have
    \begin{align*}
      \fil_{\{x\}}&(\flat(t))(m) = \flat(t)(m,x) =
      \begin{cases}
        1 & t(x) = m\neq 0\\
        0 & \text{otherwise,}
      \end{cases}\\
      \fil_{\{1\}}&(\flat(M^{(\chi_{\{x\}})}(t)))(m) =\\
      &= \flat(M^{(\chi_{\{x\}})}(t))(m,1) =
      \begin{cases}
        1 & M^{(\chi_{\{x\}})}(t)(1) = m\neq 0\\
        0 & \text{otherwise}
      \end{cases}\\
      &=
        \begin{cases}
          1 & t(x) = m \neq 0\\
          0 & \text{otherwise.}
        \end{cases}
    \end{align*}
  \item The polynomial functor $F_\Sigma$ for a signature $\Sigma$ has labels
    $A=\N$, and the map $\flat\colon F_\Sigma X\to\Bag(\N\times X)$ is given by
    \( \flat(\sigma(x_1,\ldots,x_n)) = \{(1,x_1),\ldots,(n,x_n)\}\). Since this
    $\flat$ is natural by Proposition~\ref{prop:poly-flat-nat}, the desired
    result follows from~\autoref{prop:subnat}(\ref{prop:subnat:1}).
    \qedhere
  \end{enumerate}
\end{proof}

\subsection*{Proof of \autoref{prop:prod-coprod-enc-uniform}}
\begin{proof}
  \begin{enumerate}
  \item For the coproduct of $(F_i)_{i\in I}$,
    $\flat_X\colon \coprod_{i\in I} F_iX\to\Bag(\coprod_{i\in I}A_i\times X)$ is
    defined in \autoref{const:comb-enc} as
    \[
      \flat_X\colon \coprod_{i\in I} F_i X
      \xrightarrow{\coprod_{i\in I}\flat_{X,i}}\coprod_{i\in I}\Bag(A_i\times X)
      \xrightarrow{[\Bag(\inj_i\times X)]_{i\in I}}\Bag\big(\coprod_{i\in I} A_i\times X\big).      
    \]
    We evaluate both sides of the condition for uniform encodings:
    \begin{align*}
      \fil_{\{x\}}&(\flat(\inj_i t))(\inj_j a) = \flat(\inj_i t)(\inj_j a, x)\\
      &= ([\Bag(\inj_k\times X)]_{k\in I}(\textstyle\coprod_{k\in I}\flat_k)(\inj_i t))(\inj_j a, x)\\
      &= ([\Bag(\inj_k\times X)]_{k\in I}\inj_i(\flat_i (t)))(\inj_j a, x)\\
      &=
      \begin{cases}
        0 & i \neq j\\
        \flat_i(a,x) & \text{otherwise,}
      \end{cases}\\
      \fil_{\{1\}}&(\flat((\textstyle\coprod_{k\in I}F_k)\chi_{\{x\}}(\inj_i t)))(\inj_j a) =\\
      &= \flat((\textstyle\coprod_{k\in I}F_k)\chi_{\{x\}}(\inj_i t))(\inj_j a, 1)\\
      &= \flat(\inj_i F_i\chi_{\{x\}}(t))(\inj_j a, 1)\\
      &=
      \begin{cases}
        0 & i\neq j\\
        \flat_i(F_i\chi_{\{x\}}t)(a, 1) & \text{otherwise,}
      \end{cases}
    \end{align*}
    and $\flat_i(F_i\chi_{\{x\}}t)(a,1) = \flat_i(t)(a,x)$ by the
    assumption that $\flat_i$ is uniform. Therefore, both sides agree.

  \item For the product $\prod_{i\in I}F_i$ we define
    $\flat_X\colon\prod_{i\in I}F_i X\to \Bag(\coprod_{i\in I}A_i\times X)$ in
    \autoref{const:comb-enc} as
    \[ \flat_X(t)(\inj_i(a), x) = \flat_i(\pr_i(t))(a, x). \]
    We evaluate both sides again:
    \begin{align*}
      \fil_{\{x\}}&(\flat(t))(\inj_i a) = \flat(t)(\inj_i a, x) = \flat_i(\pr_i t)(a,x)\\
      \fil_{\{1\}}&(\flat(\prod_{k\in I} F_k\chi_{\{x\}}(t)))(\inj_i a) =\\
      &=\flat(\prod_{k\in I}F_k\chi_{\{x\}}t)(\inj_i a, 1)\\
      &= \flat_i(\pr_i(\prod_{k\in I}F_k\chi_{\{x\}}(t)))(a, 1)\\
      &= \flat_i(F_i\chi_{\{x\}}(\pr_i t))(a, 1)\\
      &= \flat_i(\pr_i t)(a, x),
    \end{align*}
    where the last line uses the fact that $\flat_i$ is uniform.
    \qedhere
  \end{enumerate}
\end{proof}

\subsection*{Proof of \autoref{prop:subnat}}
In order to prove that uniform encodings are subnatural we use the following
lemma:
\begin{lemma}
  \label{5gon2emptybag}
  The following diagram commutes for all uniform encodings:
  \begin{equation*}
    \begin{tikzcd}
      F1 \ar[rr, "!"] \ar[d, "F0", swap] && 1 \ar[d, "\emptybag"]\\
      F2 \ar[r, "\flat_2"] & \Bag(A\times 2)\ar[r, "\fils{1}"] & \Bag A,
    \end{tikzcd}
  \end{equation*}
  where $0\colon 1 = \{0\} \hookrightarrow \{0, 1\} = 2$ is the
  obvious inclusion map.
\end{lemma}
\begin{proof}
  The following diagram commutes for all $n\in \N$:
  \[
    \begin{tikzcd}
      F1 \ar[r, "F\inl"]\ar[dr, "F0",swap] & F(1 + \N) \ar[d,
      "F\chi_{\{\inr n\}}"]\ar[r, "\flat_{1+\N}"] & \Bag(A\times
      (1+\N))\desctox{shift right=3ex}{d}{(Def.~\ref{def:uniform-encoding})}
      \ar[dr, "\fils{\inr n}"]\\
      & F2\ar[r, "\flat_2", swap] & \Bag(A\times 2)\ar[r, "\fils{1}", swap] & \Bag(A).
    \end{tikzcd}
  \]
  Let $t\in F1$ and, for the sake of contradiction, suppose that
  $\fils{1}(\flat_2(F0(t)))$ is nonempty and contains the element
  $a$. Then, by the above diagram we have $a\in\fils{\inr n}(\flat_{1+\N}(F\!\inl\,(t)))$ and therefore
  \[ (a, \inr n) \in \flat_{1+\N}(F\!\inl\,(t))\quad\quad\text{for all }n\in \N.\]
  However, this contradicts the finiteness of the bag $\flat_{1+\N}(F\!\inl\,(t))$.
\end{proof}

\noindent We are now ready to prove the main proposition:
\begin{proof}[Proof of \autoref{prop:subnat}]
  \begin{enumerate}
  \item Given an encoding $\flat_X\colon FX \to \Bag(A \times X)$ which is natural in $X$,
    we have the following commutative diagram:
    \[
      \begin{tikzcd}[row sep = 5]
        FX
        \ar[r, "\flat_X"]
        \ar[dd, "F\chi_{\{x\}}", swap]
        &
        \Bag(A\times X)
        \ar[dd, "\Bag(A\times\chi_{\{x\}})",swap]
        \ar[dr, near start, "\fil_{\{x\}}"]\\
        & & \Bag(A),\\
        F2 \ar[r, "\flat_Y"]
        &
        \Bag(A\times 2)
        \ar[ur, near start, "\fil_{\{1\}}",swap]
      \end{tikzcd}      
    \]
    Indeed, the left-hand square commutes due to the naturality of $\flat$ and
    the right-hand triangle commutes by
    \autoref{lem:fils-properties}(\ref{lem:filsingle2}).
    
  \item Let $\flat_X\colon FX\to\Bag(A\times X)$ be a uniform encoding. First we show
  that the family
  $\big(\Bag(A\times Y)\xrightarrow{\fils{y}}\Bag(A)\big)_{y\in Y}$ is jointly
  monic. Indeed, recall from \autoref{lem:fils-properties} that the morphism
  $\langle\fils{y}\rangle_{y\in Y}$ is equal to
  $\group\colon\Bag(A\times Y)\to(\Bag A)^{(Y)}$, which is an isomorphism,
  whence a split mono. It therefore suffices to prove that the following diagram
  commutes for all $y\in Y$ and all monomorphisms $m\colon X\monoto Y$:
  \begin{equation}\label{eq:subnat-proof}
    \begin{tikzcd}
      FX \ar[r, "\flat_X"]\ar[d, "Fm", swap] & \Bag(A\times X)\ar[r, "\Bag(A\times m)"] & \Bag(A\times Y)\ar[d, "\fils{y}"]\\
      FY \ar[r, "\flat_Y"] & \Bag(A\times Y)\ar[r, "\fils{y}"] & \Bag A
    \end{tikzcd}
  \end{equation}

  \noindent
  We distinguish two cases:
  \begin{enumerate}
  \item If $y\in m[X]$, equivalently, $y=m(x)$ for an $x\in X$, the
    following diagram commutes:
    \begin{equation*}
      \begin{tikzcd}
        FX \ar[rr, "\flat_X"]\ar[d, "Fm", swap]\ar[dr, "F\chi_{\{x\}}"]
        \descto{drrr}{(Def.~\ref{def:uniform-encoding})}
        && \Bag(A\times X)\ar[r, "\Bag(A\times m)"]\ar[dr, "\fils{x}"] & \Bag(A\times Y) \ar[d, "\fils{m(x)}"]\\
        FY \ar[r, "F\chi_{\{m(x)\}}", swap]\ar[dr, "\flat_Y", swap, bend right] & F2 \ar[r, "\flat_2"] & \Bag(A\times 2)\ar[r, "\fils{1}"] & \Bag(A)\\
        & \Bag(A\times Y)\ar[urr, "\fils{m(x)}", swap, bend right=10]
        \descto{ur}{(Def.~\ref{def:uniform-encoding})}
      \end{tikzcd}
    \end{equation*}
    Therefore, \eqref{eq:subnat-proof} commutes for $y=m(x)\in m[X]$.
  \item If $y\in (Y\setminus m[X])$, equivalently, $\chi_{\{y\}}\cdot m
    = 0!$, then the following diagram commutes:
    \begin{equation*}
      \begin{tikzcd}
        FX \ar[rr, "\flat_X"]\ar[dd, "Fm", swap]\ar[dr, "F!"] &&
        \Bag(A\times X)\ar[d, "!"]\ar[r, "\Bag(A\times m)"]
        \descto{dr}{(\(y\not \in m[X]\))}
        & \Bag(A\times Y)\ar[dd, "\fils{y}"]\\
        & F1 \ar[r, "!"]\ar[d, "F0"]
        \descto[xshift=-8mm]{drr}{(Lem.~\ref{5gon2emptybag})}
        & 1 \ar[rd, "\emptybag"] & {}\\
        FY \ar[r, "F\chi_{\{y\}}"]\ar[dr, "\flat_Y", bend right, swap]
        & F2 \ar[r, "\flat_2"]
        \descto{d}{(Def.~\ref{def:uniform-encoding})}
        & \Bag(A\times 2)\ar[r, "\fils{1}", near start] & \Bag A\\
        & \Bag(A\times Y)\ar[rru, "\fils{y}", bend right=10, swap]
      \end{tikzcd}
    \end{equation*}
    Therefore, \eqref{eq:subnat-proof} also commutes for $y\not\in m[X]$.\qedhere
  \end{enumerate}
\end{enumerate}
\end{proof}

\subsection{Proofs for \autoref{sec:simpl}}
\subsection*{Proof of \autoref{L:merge-enc}}
\begin{proof}%
  To see this, instantiate~\eqref{eq:merge-axiom} for $S = \{x\}$ and compare it
  with the diagram~\eqref{eq:uniform-encoding} of \autoref{def:uniform-encoding}:
  \[
    \begin{tikzcd}[column sep=19mm]
      FX \ar[r, "\fil_{\{x\}}\cdot\,\flat_X"]
      \ar[d, "F\chi_{\{x\}}",swap]
      \ar[rr, shiftarr={yshift=8mm},"\flat_X"]
      \descto[xshift=-2mm]{drr}{\eqref{eq:merge-axiom}}
      & \Bag(A)\ar[dr, "\merge"{description}]
      \descto{r}{\ref{L:merge-enc}(\ref{L:merge-enc:1})}
      & \Bag(A\times X)
      \ar[d, "\fil_{\{x\}}"]
      \\
      F2 \ar[rr, "\fil_{\{1\}}\cdot\,\flat_2", swap]
      & & \Bag(A)
    \end{tikzcd}
  \]
  We see that the upper inner part of the diagram commutes if and only if the
  outside does. This establishes the desired equivalence.
\end{proof}

\subsection*{Proof of \autoref{lem:merge-id-enc-nat}}
We first establish the following easy lemma:
\begin{lemma}\label{lem:revimg}
  For every map $f\colon X\to Y$ and subset $S\subseteq Y$ we have
  \[
    \begin{tikzcd}
      X
      \arrow{rr}{f}
      \arrow{rd}[swap]{\chi_{f^{-1}[S]}}
      &&
      Y
      \arrow{ld}{\chi_S}
      \\
      &
      2
    \end{tikzcd}
  \]
\end{lemma}
\begin{proof}
  Indeed, we have
  \[
    \chi_S(f(x)) = 1 \Leftrightarrow f(x)\in S \Leftrightarrow x\in f^{-1}[S]
    \Leftrightarrow \chi_{f^{-1}[S]}(x) = 1.\qedhere
  \]
\end{proof}
\begin{proof}[Proof of \autoref{lem:merge-id-enc-nat}]

  \ref{item:enc-nat-merge-id} $\Rightarrow$ \ref{item:merge-id-enc-nat}: The
  required axiom~\eqref{eq:merge-axiom} for $\merge = \id$ follows by combining
  the given naturality of $\flat_X$ with Equation~\eqref{eq:filsingle2}:
  \[
    \begin{tikzcd}[row sep=5]
      FX
      \ar[r, "\flat_X"]
      \ar[dd, "F\chi_S",swap]
      &
      \Bag(A\times X)
      \ar[dd, "\Bag(A\times\chi_S)",swap]
      \ar[dr, "\fil_S"]
      \\
      &&
      B(A)
      \\
      F2
      \ar[r, "\flat_2"]
      &
      \Bag(A\times 2)
      \ar[ur, "\fils{1}",swap]
    \end{tikzcd}
  \]

  \medskip\noindent
  \ref{item:merge-id-enc-nat} $\Rightarrow$ \ref{item:enc-nat-merge-id}:
  Suppose $\merge = \id$ is a minimization interface. Then the
  axiom~\eqref{eq:merge-axiom} simplifies as follows: 
  \begin{equation}\label{eq:merge-id-axiom}
    \begin{tikzcd}[row sep=5]
      FX \ar[r, "\flat_X"]\ar[dd, "F\chi_S",swap] & \Bag(A\times X) \ar[dr, "\fil_S"]\\
      & & B(A).\\
      F2 \ar[r, "\flat_2"] & \Bag(A\times 2) \ar[ur, "\fils{1}",swap]
    \end{tikzcd}
  \end{equation}
  In order to show that $\flat$ is a natural transformation we use that the family
  $\big(\Bag(A\times Y)\xrightarrow{\fils{y}}\Bag(A)\big)_{y\in Y}$ is jointly
  monic. Hence, it suffices to prove that the following diagram commutes for all
  functions $f\colon X \to Y$ and $y \in Y$:
  \begin{equation}\label{eq:merge-nat-proof}
    \begin{tikzcd}[column sep = 35]
      FX \ar[r, "\flat_X"]\ar[d, "Ff", swap] & \Bag(A\times X)\ar[r, "\Bag(A\times f)"] & \Bag(A\times Y)\ar[d, "\fils{y}"]\\
      FY \ar[r, "\flat_Y"] & \Bag(A\times Y)\ar[r, "\fils{y}"] & \Bag A
    \end{tikzcd}
  \end{equation}
  \noindent
  Indeed, let $y\in Y$ and $S\subseteq X$ be the inverse image of $y$ under
  $f$: $S = f^{-1}[y]$. Then the following diagram commutes:
  \[
    \begin{tikzcd}[column sep = 35]
      FX \ar[rr, "\flat_X"]\ar[d, "Ff", swap]\ar[dr, "F\chi_S"]
      \desctox[pos=.25]{shift right=1em}{dr}{\ref{lem:revimg}}
      \descto{drrr}{\eqref{eq:merge-id-axiom}}
      && \Bag(A\times X)\ar[r, "\Bag(A\times f)"]\ar[dr, "\fil_S"] & \Bag(A\times Y) \ar[d, "\fils{y}"]\\
      FY \ar[r, "F\chi_{\{y\}}", swap]\ar[dr, "\flat_Y", swap, bend right] & F2 \ar[r, "\flat_2"]
      & \Bag(A\times 2)\ar[r, "\fils{1}"] & \Bag(A)\\
      & \Bag(A\times Y)\ar[urr, "\fils{y}", swap, bend right=10]
      \descto{ur}{\eqref{eq:merge-id-axiom} for $S=\{y\}$}
    \end{tikzcd}
    \]
  Therefore,~\eqref{eq:merge-nat-proof} commutes for all $f\colon X \to Y$ and
  $y \in Y$ as desired.
\end{proof}

\subsection*{Proof of \autoref{prop:all-merges-go}}

\begin{proof}
\begin{enumerate}
\item For the finite powerset functor $\Powf(-)$, with $A = 1$, we define \merge{} by
  \[ \merge(\ell)(*) = \min(1, \ell(*)). \]
  To show that the axiom holds, we calculate both sides:
  \begin{align*}
    \fil_{\{1\}}&(\flat(\Powf\chi_S(t)))(*) = \flat(\Powf\chi_S(t))(*, 1)\\
        &= \begin{cases}
              1 & 1\in \Powf\chi_S (t)\\
              0 & \text{otherwise}
            \end{cases}
         =
            \begin{cases}
              1 & S\cap t \neq\emptyset\\
              0 & \text{otherwise}
            \end{cases}\\[1em]
    \merge&(\fil_S(\flat(t)))(*) = \min(1, \fil_S(\flat(t))(*))\\
                  &= \min(1, \sum_{x\in S}\underbrace{\flat(t)(*, x)}_{x\in t\Rightarrow 1,\text{ else }0})\\
                  &=
                    \begin{cases}
                      1 & S\cap t \neq\emptyset\\
                      0 & \text{otherwise.}
                    \end{cases}
  \end{align*}
  This $\merge{}$ can be implemented in constant time, since it just needs to
  check if its input bag is empty and return one of two possible constants,
  depending on that result.
\item For monoid-valued functors $M^{(-)}$ with $A = M$, \merge{} is defined as
  \[
    \merge(\ell) =
    \begin{cases}
      \mbraces{ \Sigma\ell } & \Sigma \ell\neq 0\\
      \emptybag & \text{otherwise.}
    \end{cases}
  \]
  To show that this fulfils the required property, we first need the following facts:
  \begin{align*}
    (M^{(\chi_S)}t)(1) &= \sum_{\mathclap{\substack{x\in X\\\chi_S(x) = 1}}} t(x) = \sum_{x\in S}t(x),\\
    \intertext{and}
    \Sigma(\fil_S(\flat(t))) &= \sum_{m\in M}m\cdot (\fil_S(\flat(t)))(m)\\
                       &=\sum_{m\in M}m\cdot\left(\sum_{x\in S}\flat(t)(m,x) \right)\\
                       &=\sum_{x\in S}\sum_{m\in M} m\cdot\underbrace{\flat(t)(m,x)}_{\substack{1\text{ if }t(x)=m\neq 0\\\text{else }0}}\\
                       &=\sum_{x\in S}t(x).
  \end{align*}
  Now we have that
  \begin{align*}
    \fil_{\{1\}}&(\flat(M^{(\chi_S)}t))(m) = \flat(M^{(\chi_S)}t)(m, 1)\\
                &=
                  \begin{cases}
                    1 & \sum_{x\in S}t(x) = m\neq 0\\
                    0 & \text{otherwise}
                  \end{cases}\\
    \intertext{and also}
    \merge&(\fil_S(\flat(t)))(m) =
            \begin{cases}
              1 & \Sigma(\fil_S(\flat(t))) = m\neq 0\\
              0 & \text{otherwise}
            \end{cases}\\
                &=
                  \begin{cases}
                    1 & \sum_{x\in S}t(x) = m\\
                    0 & \text{otherwise.}
                  \end{cases}
  \end{align*}
  Since this \merge{} has to sum up the monoid elements in its input bag, it
  runs in linear time in the size of that bag, provided that addition of monoid
  elements is a constant-time operation.
\item For the polynomial functor $F_\Sigma$, the encoding
  $\flat\colon F_\Sigma X\to\Bag(\N\times X)$ is already natural (see
  \autoref{prop:poly-flat-nat}). Thus, $\merge = \id$ is a minimization
  interface by \autoref{lem:merge-id-enc-nat} with constant run-time.%
  \qedhere
\end{enumerate}
\end{proof}

\subsection*{Proof of \autoref{thm:merge-thm}}
We first prove the following technical proposition about merge:
\begin{proposition}\label{prop:merge-empty}
  Suppose $F$ is not the constant empty set functor ($CX=\emptyset$, $F\neq C$)  and is equipped with a subnatural encoding and a
  minimization interface $\merge$. Then we have
  $\merge(\mbraces{}) = \mbraces{}$.
\end{proposition}
\begin{proof}%
  Consider the diagram for the injective $\chi_\emptyset\colon 1\rightarrowtail 2$:
  \[
    \begin{tikzcd}
      F1
      \arrow{r}{\flat_1}
      \arrow{d}[swap]{F\chi_\emptyset}
      & \Bag(A\times 1)
      \arrow{r}{\fil_\emptyset}
      \arrow{d}{\Bag(A\times\chi_\emptyset)}
      & \Bag A
      \arrow{d}{\merge}
      \\
      F2
      \arrow{r}{\flat_2}
      & \Bag(A\times 2)
      \arrow{r}{\fil_{\{1\}}}
      & \Bag A
    \end{tikzcd}
  \]
  Note that $F1$ is nonempty for all $\Set$-functors except for the constant
  empty set functor, which is excluded by assumption. Hence, there is some $x\in
  F1$ and we have:
  \begin{align*}
    \merge(\mbraces{})
    &= \merge (\fil_\emptyset(\flat_1(x)))
      \tag{$\fil_\emptyset$ always returns $\mbraces{}$}
    \\
    &= \fil_{\{1\}}(\flat_2(\chi_\emptyset(x)))
      \tag{$\merge$ axiom}
    \\
    &= \fil_{\{1\}}(\Bag(A\times\chi_\emptyset)(\flat_1(x))))
      \tag{$\flat$ subnatural}
    \\
    &= \mbraces{} \tag{$1$ is not in the image of $\chi_\emptyset$}
  \end{align*}
  as desired.
\end{proof}

\noindent We now proceed to prove \autoref{thm:merge-thm}:
\begin{proof}
  First, observe that for every function $f\colon X \to Y$ the following
  squares commute:
  \begin{equation}\label{eq:extensionality}
    \begin{tikzcd}
      X^Z \ar[d, "f^Z", swap] \ar[r, "\ev(z)"] & X \ar[d, "f"]\\
      Y^Z \ar[r, "\ev(z)"] & Y
    \end{tikzcd}
    \qquad
    \text{for every $z \in Z$.}
  \end{equation}
  We verify that the following diagram commutes for every $y\in Y$, where we define $S = q^{-1}[y]$:
  \[
    \begin{tikzcd}
      FX \ar[ddd, "Fq",swap] \ar[r, "\flat_X"]\ar[drd, "F\chi_S"]\desctox{bend
        left}{ddd}{\ref{lem:revimg}}\descto{ddrrr}{\eqref{eq:merge-axiom}} & \Bag(A\times X)\ar[rrd, "\fil_S",swap]\desctox[near end]{bend left=20}{rrd}{\eqref{eq:filsingle2}}\ar[rr, "\Bag(A\times q)"] & & \Bag(A\times Y)\ar[d, "\fil_{\{y\}}"]\ar[r, "\group"] & {\Bag A}^{(Y)} \ar[dl, "\ev(y)", rounded corners, to path={|-  (\tikztotarget) [near start]\tikztonodes}]\descto{dl}{\rlap{\!\!\!\eqref{eq:filsingle}}}\\
      & &  & \Bag A\ar[d, "\merge"]\\
      & F2 \ar[r, "\flat_2"]\descto{drr}{\eqref{eq:uniform-encoding}} & \Bag(A\times 2)\ar[r, "\fil_{\{1\}}"] & \Bag A  & \\
      FY \ar[ur, "F\chi_{\{y\}}"]\ar[rrr, "\flat_Y"] &&& \Bag(A\times
      Y)\ar[u, "\fil_{\{y\}}",swap]\ar[r, "\group"] & {\Bag
        A}^{(Y)} \ar[ul, "\ev(y)", swap,rounded corners, to path={|-
        (\tikztotarget) [near start]\tikztonodes}]\descto{ul}{\rlap{\eqref{eq:filsingle}}}
    \end{tikzcd}
  \]
  Instantiating~\eqref{eq:extensionality} for $f = \merge$ yields
  $\ev(y) \cdot \merge^Y = \merge\cdot\ev(y)$. By virtue of
  \autoref{prop:merge-empty} and \autoref{ass:simpl}, \merge{} preserves empty bags. Hence, the function
  $\merge^{Y}\colon \Bag(A)^{Y}\to\Bag(A)^{Y}$ defined as
  $f\mapsto \merge\cdot f$ preserves finite support and therefore restricts to
  the monoid-valued functor as
  $\merge^{(Y)}\colon\Bag(A)^{(Y)}\to\Bag(A)^{(Y)}$. Therefore, the outside of
  the diagram together with the fact that $(\ev(y))_{y\in Y}$ is a jointly
  injective family implies
  \[
    \begin{tikzcd}[column sep = 35]
      FX \ar[r, "\flat_X"]\ar[d, "Fq",swap] & \Bag(A\times X)\ar[r, "\Bag(A\times q)"] & \Bag(A\times Y)\ar[r, "\group"] & \Bag(A)^{(Y)} \ar[d, "\merge^{(Y)}"]\\
      FY \ar[rr, "\flat_Y"] & & \Bag(A\times Y) \ar[r, "\group"] & \Bag(A)^{(Y)}.\\
    \end{tikzcd}
  \]
  Post-composition with $\ungroup$ and the application of equation
  \eqref{eq:group-ungroup} now yields the desired result.
\end{proof}

\subsection*{Proof of \autoref{constr:correct:linear}}%
\newcommand{\yidx}{\operatorname{idx}}
\begin{proof} \leavevmode
  (1) Correctness. 
  Combining that $q$ is a coalgebra homomorphism with \autoref{thm:merge-thm}
  yields the following diagram, whose commutativity we discuss next:
  \[
    \begin{tikzcd}[column sep = 35, row sep=14mm]
      X
      \arrow{r}{c}
      \arrow[->>]{d}[swap]{q}
      \descto{dr}{$q$ coalgebra\\ morphism}
      &[8mm] FX \ar[r, "\flat_X"]\ar[d, "Fq",swap]
      \descto{drrr}{\autoref{thm:merge-thm}}
        & \Bag(A\times X)\ar[r, "\Bag(A\times q)"]
        & \Bag(A\times Y)\ar[r, "\group"] 
        & \Bag(A)^{(Y)} \ar[d, "\merge^{(Y)}"]\\
      Y
      \arrow{r}{d}
      \arrow[shiftarr={yshift=-7mm}]{rrr}[swap]{e}
      &FY \ar[rr, "\flat_Y"]
        &
        & \Bag(A\times Y)
        & \Bag(A)^{(Y)} \ar[l, "\ungroup",swap]
    \end{tikzcd}
  \]
  The two rectangles commute, and the outside of the diagram commutes by \autoref{cons:main}.
  Hence, $e\cdot q = \flat_Y\cdot d \cdot q$. Since $q$ is surjective, we have
  $e = \flat_Y\cdot d$ as desired.

  \medskip
  \noindent
  (2)~Runtime.
    For the implementation of \autoref{cons:main}, assume that the encoded input coalgebra
    $\flat_X\cdot c\colon X\to \Bag(A\times X)$ is given as adjacency lists
    and that the quotient map $q\colon X\epito Y$ is given as a partition on
    $X$. Such a partition is represented as an assignment $q'\colon X\epito
    \{0,\ldots, |Y|-1\}$ which sends an element of $X$ to the number of its
    block and which can be evaluated in $\CO(1)$ (e.g.~the refinable partition
    structure~\cite{ValmariLehtinen08} represents partitions in such a way and
    is in fact used by the coalgebraic algorithm~\cite{coparFM19}); in other
    words, we implicitly use the bijection $Y = \{0,\ldots,|Y|-1\}$. We now
    compute the composition
    \[
    \underbrace{\ungroup\cdot\,\underbrace{\merge{{}^{(Y)}}\cdot \underbrace{\group \cdot \underbrace{\Bag(A\times
    q)\cdot \underbrace{\flat_X\cdot c}_{s_0}}_{s_1}}_{s_2}}_{s_3}}_{s_4}\colon~~ X\longrightarrow \Bag(A\times Y)
    \]
    from \autoref{cons:main} step by step:
    \begin{enumerate}
    \item $s_0 := \flat_X\cdot c\colon X\to \Bag(A\times X)$ is the given input,
      encoded using adjacency lists, i.e.~as an array of size $|X|$ whose entries are lists of
      elements from $A\times X$. We denote its size by
      \[
        m := |X| + \sum_{x\in X}|\flat_X(c(x))|.
      \]
    \item For $s_1 := \Bag(A\times q)\cdot s_0\colon X\to \Bag(A\times Y)$,
      we iterate over all edges in the adjacency lists and replace every
      right-hand side $x\in X$ of an edge by $q(x) \in Y$. This takes $\CO(m)$
      time ($\CO(1)$ time for each of the $m$ entries).
    \item For $s_2 := \group\cdot\,s_1\colon X\to \Bag(A)^{(Y)}$, we represent a
      map $t\in \Bag(A)^{(Y)}$ as a list of pairs $(y,t(y)) \in Y\times \Bag(A)$ with $t(y)$
      non-empty and compute this list for all $x\in X$ as follows. \twnote{} Allocate an array $\yidx\colon Y\to \mathbb{Z}$ (initially $-1$ everywhere)
      and then do the following for every $x\in X$:
      \smnote{}
      \begin{enumerate}
      \item Allocate an array $p\colon \N\to Y\times \Bag(A)$ of size $|s_1(x)|$
        and initialize an integer $i:=0$ (intuitively, $i$ is the index of the
        first unused cell in $p$).
      \item For every $(a,y) \in s_1(x)$, we distinguish whether we have seen
        $y$ in $s_1(x)$ before:
        \begin{itemize}
        \item If $\yidx(y) < 0$, then it is the first time we see $y$ in $s_1(x)$.
          Thus put $\yidx(y) := i$, increment $i := i+1$, and define $p(\yidx(y)) := (y,
          \mbraces{a})$.
        \item If $\yidx(y) \ge 0$, then we have seen $y$ before and simply append
          $a$ to the second component of $p(\yidx(y))$.
        \end{itemize}
      \item For every $(y,\ell)$ in the first $i$ entries of $p$, put $\yidx(y) =
        -1$. (Thus, $\yidx$ is again $-1$ everywhere.)
      \item Let $s_2(x)$ be the first $i$ entries of $p$
      \end{enumerate}
      For $x\in X$ each of the above steps runs in $\CO(|s_1(x)|)$, thus
      doing these for all $x\in X$ runs in $\CO(m)$ in total.

    \item For $s_3 := \merge^{(Y)}\cdot s_2\colon X\to \Bag(A)^{(Y)}$, apply
      $\merge\colon \Bag(A)\to \Bag(A)$ to every bag in the list $s_2(x)\in \Bag(A)^{(Y)}$
      (we have represented $s_2(x)$ as a list of elements from $Y\times \Bag(A)$
      in the definition of $s_2$). Since by assumption, $\merge$ runs in
      linear time, the present step runs in $\CO(m)$ time\smnote{} and moreover the size of
      the resulting $s_3$ is still of size $\CO(m)$.

    \item For $s_4 := \ungroup\cdot s_3\colon X\to \Bag(A\times Y)$, first note
      that for every $x\in X$, the bag $s_3(x) \in \Bag(A)^{(Y)}$ is represented
      by a list of elements of $Y\times \Bag(A)$, i.e.~every $\ell\in s_3(x)$ is
      of type $\ell \in Y\times \Bag(A)$, thus we define $s_4$ as the following
      multiset-comprehension:
      \[
        s_4(x) := \mbraces{ (a,y)\mid (y, t) \in s_3(x), a\in t }
        \in \Bag(A\times Y).
      \]
      This is computed in time $|s_3(x)|$ for every $x\in X$ and thus $s_4$ can
      be computed in time~$\CO(m)$.
    \end{enumerate}
    Finally, for the definition of $e\colon Y\to \Bag(A\times Y)$, we allocate
    $|Y|$ new adjacency lists, all of them empty initially. Then, for every $x \in X$, we put 
    \(
      e(q(x)) := s_4(x)
    \)
    if $e(q(x))$ is empty (and skip otherwise). By the well-definedness of
    \autoref{cons:main} it does not matter which~$x\in X$ defines the outgoing
    edges of $q(x)\in Y$. This takes $|X| < m$ time. Thus, all steps
    $s_1,\ldots,s_4$ and the final definition of $e$ take $\CO(m)$ time in total.
    \qedhere
\end{proof}
\subsection*{Proof of \autoref{prop:merge-for-prod}}
We first note a few technicalities before proceeding to the proof of
\autoref{prop:merge-for-prod}.
\begin{remark}\label{rem:filter-facts}
  \begin{enumerate}
  \item We observe that for every $i\in I$, we have
    \begin{equation}
      \label{eq:prod-merge-filter}
      \filter_i(\merge(t))(a) = \merge(t)(\inj_i a) = \merge_i(\filter_i(t))(a).
    \end{equation}

  \item In order to show that $\merge{}$ in \autoref{C:merge-prod}
    indeed constitutes a lawful minimization interface, we use a
    different, but equivalent, definition of $\flat$ for
    $\prod_{i\in I}F_i$:
    \begin{equation}
      \label{eq:new-enc-for-prod}
      \flat' = \prodi F_i X\xrightarrow{\Pi_{i\in I}\flat_i}
        \prodi\Bag(A_i\times X) \xrightarrow{\concat}
        \Bag(\coprodi A_i\times X),
    \end{equation}
    with $\concat$ given by
    $\concat(t)(\inj_i a, x) = \pr_i(t)(a, x)$. This is indeed
    equivalent to the original definition:
    \begin{align*}
      \flat'(t)(\inj_i a, x) &= \concat(\prodi[j]\flat_j(t))(\inj_i a, x)\\
      &= \pr_i(\prodi[j]\flat_j(t))(a, x) = \flat_i(\pr_i(t))(a, x)\\
      &= \flat(t)(\inj_i a, x).
    \end{align*}
    We also need another auxiliary definition similar to
    $\concat{}$
    \begin{align}
      &\concat'\colon \prodi\Bag(A_i)\to\Bag(\coprodi A_i) \nonumber\\
      &\concat'(t)(\inj_i a) = \pr_i(t)(a),
        \label{defConcatPrime}
    \end{align}
    for which we observe the following properties:

    \begin{equation}\label{eq:concat-filter}
      \begin{tikzcd}[column sep = 30]
        \prodi[j]\Bag(A_j)\ar[r, "\concat'"]\ar[rd, "\pr_i",swap] & \Bag(\coprodi[j]A_j)\ar[d, "\filter_i"]\\
        & \Bag(A_i)
      \end{tikzcd}
    \end{equation}
    and
    \begin{equation}\label{eq:concat-concat}
      \begin{tikzcd}[column sep = 30]
        \prodi\Bag(A_i\times X)\ar[r, "\concat"]\ar[d, "\prodi\fil_S"] & \Bag(\coprodi A_i\times X)\ar[d, "\fil_S"]\\
        \prodi\Bag(A_i)\ar[r, "\concat'",swap] & \Bag(\coprodi A_i).
      \end{tikzcd}
    \end{equation}
    Indeed, we have
    \begin{align*}
      \fil_S(\concat(t))(\inj_i a) &= \sum_{x\in S}\concat(t)(\inj_i a, x)\\
      &=\sum_{x\in S}\pr_i(t)(a,x) = \fil_S(\pr_i(t))(a)\\
      &=\pr_i(\prodi[j]\fil_S(t))(a) \\
      &= \concat'(\prodi[j]\fil_S(t))(\inj_i, a)
    \end{align*}
    and
    \[
      \filter_i(\concat'(t))(a) = \concat'(t)(\inj_i a) = \pr_i(t)(a).
    \]
    
  \item\label{rem:filter-facts:3} The function $\filter_i$ behaves as expected
    when injecting all elements of a bag into a coproduct and then immediately
    filtering this bag. Speci\-fi\-cally, we have that
    \begin{equation}
      \filter_i\cdot\Bag(\inj_j) =
      \begin{cases}
        \id & i = j,\\
        \emptybag! & i \neq j.
      \end{cases}
    \end{equation}
  \end{enumerate}
\end{remark}

\noindent We are now ready to prove the main proposition:
\begin{proof}[Proof of \autoref{prop:merge-for-prod}]
  \begin{enumerate}
  \item For the product functor $\prodi F_i$, the following diagram commutes for all $i\in I$:
  \[
    \begin{turn}{90}
    \begin{tikzcd}
      &&& \Bag(\coprodi[j]A_j\times X)\ar[dr, "\fil_S"]\\
      \prodi[j]F_jX \ar[rr, "{\prodi[j]\flat_j}"]\ar[dr, "\pr_i"]\ar[ddd, "{\prodi[j]F_j\chi_S}", swap]\ar[urrr, "\flat'",bend left]\desctox{shift left=5ex}{ddd}{\scriptsize Naturality\\ of $\pr_i$}\descto{rrrd}{\scriptsize Naturality of $\pr_i$}\desctox{shift left=4ex}{rrru}{Def.\ of \(\flat'\)~\eqref{eq:new-enc-for-prod}} && \prodi[j]\Bag(A_j\times X)\ar[ur, "\concat"]\ar[r, "{\prodi[j]\fil_S}",swap]\desctox{shift left=4ex}{rr}{\eqref{eq:concat-concat}} & \prodi[j]\Bag(A_j)\ar[d, "\pr_i"]\ar[r, "\concat'"] & \Bag(\coprodi[j] A_j)\ar[dl, "\filter_i"]\ar[dd, "\merge"]\desctox{shift right=8ex}{ddd}{\eqref{eq:prod-merge-filter}}\\
      & F_iX \ar[d, "F_i\chi_S", swap]\ar[r, "\flat_i"]\descto{drr}{Axiom $\merge_i$} & \Bag(A_i\times X)\ar[r, "\fil_S"] & \Bag(A_i)\ar[d, "\merge_i"]\desctox[near start]{shift left=2ex}{ur}{\eqref{eq:concat-filter}}\\
      & F_i2 \ar[r, "\flat_i", swap] & \Bag(A_i\times 2)\ar[r, "\fil_{\{1\}}",
      swap] & \Bag(A_i)\desctox[near start]{shift
        right=2ex}{dr}{\eqref{eq:concat-filter}}
      & \Bag(\coprodi[j]A_j) \arrow{l}[swap]{\filter_i}
      \\
      \prodi[j]F_j2 \ar[ur, "\pr_i"]\ar[rr, "{\prodi[j]\flat_j}",swap]\ar[drrr, "\flat'", swap, bend right]\descto{rrru}{\scriptsize Naturality of $\pr_i$}\desctox{shift right=4ex}{rrrd}{Def.\ of \(\flat'\)~\eqref{eq:new-enc-for-prod}} && \prodi[j]\Bag(A_j\times 2)\ar[dr, "\concat"]\ar[r, "{\prodi[j]\fil_{\{1\}}}"]\desctox{shift right=4ex}{rr}{\eqref{eq:concat-concat}} & \prodi[j]\Bag(A_j)\ar[u, "\pr_i"]\ar[r, "\concat'", swap]& \Bag(\coprodi[j] A_j)\ar[ul, "\filter_i",swap]\\
      &&& \Bag(\coprodi[j]A_j\times 2)\ar[ur, "\fil_{\{1\}}",swap]\\
    \end{tikzcd}
  \end{turn}
\]
  Observe that for any two $f,g\in\Bag(\coprodi[j]A_j)$ with $f\neq g$, there exists a $j\in I$ such that $\filter_j(f) \neq \filter_j(g)$: Let w.l.o.g.\ be $x = \inj_i a\in \coprodi[j]A_j$ such that $f(x)\neq g(x)$. Then we have
  \[
    \filter_i(f)(a) = f(\inj_i a) = f(x) \neq g(x) = g(\inj_i a) = \filter_i(g)(a).
  \]
  Hence, the family $(\filter_i)_{i\in I}$ is a point-separating source and therefore jointly monic.

  The desired equation $\merge\cdot\fil_S\cdot\flat' =
  \fil_{\{1\}}\cdot\flat'\cdot\prodi[j]F_j\chi_S$ thus follows from the diagram
  above.

\item For the coproduct functor $\coprodi F_i$, we assume without loss of
  generality that $F_i1 \neq \emptyset$ for all $i \in I$ because summands which
  are constantly $\emptyset$ may be omitted from the coproduct without changing
  it.

  We need to show
  \[ \merge(\fil_S(\flat(\inj_i t)))(\inj_j a) = \fils{1}(\flat((\coprodi[k]F_k\chi_S)(\inj_i t)))(\inj_j a)
  \]
  for every $\inj_i t\in\coprodi[k]F_kX$ and $\inj_j a\in\coprodi[k]A_k$.

  We calculate as follows:
  \allowdisplaybreaks
  \begin{align*}
    \merge(&\fil_S(\flat(\inj_i t)))(\inj_j a)\\
           &= \merge_j(\filter_j(\fil_S(\flat(\inj_i t))))(a)\tag*{Def.\ of $\merge$}\\
           &= \merge_j(\filter_j(\fil_S([\Bag(\inj_k\times X)]_{k\in I}((\coprodi[k]\flat_k)(\inj_i t)))))(a)\tag*{Def.\ of $\flat$}\\
           &= \merge_j(\filter_j(\fil_S([\Bag(\inj_k\times X)]_{k\in I}(\inj_i\flat_i(t)))))(a)\tag*{~~\({\textstyle\coprod}\flat_k\cdot\inj_i=\inj_i\cdot\flat_i\)}\\
           &= \merge_j(\filter_j(\fil_S(\Bag(\inj_i\times X)(\flat_i(t)))))(a)\tag*{\([f_k]\cdot\inj_i=f_i\)}\\
           &= \merge_j(\filter_j(\Bag(\inj_i)(\fil_S(\flat_i(t)))))(a)\tag*{\itemref{lem:fils-properties}{prop:fils-natural}}\\
    \intertext{From here we  proceed by case distinction. If $i=j$, we have}
    \merge_i&(\filter_i(\Bag(\inj_i)(\fil_S(\flat_i(t)))))(a)\\
           &= \merge_i(\fil_S(\flat_i(t)))(a)  \tag*{\autoref{rem:filter-facts}(\ref{rem:filter-facts:3})}\\
           &= \fils{1}(\flat_i(F_i\chi_S(t)))(a)\tag*{Axiom of $\merge_i$}\\
           &= \filter_i(\Bag(\inj_i)(\fils{1}(\flat_i(F_i\chi_S(t)))))(a) \tag*{\autoref{rem:filter-facts}(\ref{rem:filter-facts:3})}\\
           &= \Bag(\inj_i)(\fils{1}(\flat_i(F_i\chi_S(t))))(\inj_i a) \tag*{Def.\ of $\filter_i$}\\
           &= \fils{1}(\Bag(\inj_i\times X)(\flat_i(F_i\chi_S(t))))(\inj_i a)\tag*{\itemref{lem:fils-properties}{prop:fils-natural}}\\
           &= \fils{1}([\Bag(\inj_k\times X)]_{k\in I}(\coprodi[k]\flat_k(\coprodi[k]F_k\chi_S(\inj_i t))))(\inj_i a) \tag*{UMP of ${\textstyle\coprod}$}\\
           &= \fils{1}(\flat(\coprodi[k]F_k\chi_S))(\inj_i a)\tag*{Def.\ of $\flat$}\\
    \intertext{In the second case, $i\neq j$, we have}
    \merge_j&(\filter_j(\Bag(\inj_i)(\fil_S(\flat_i(t)))))(a)\\
           &=\merge_j(\emptybag)(a) \tag*{\autoref{rem:filter-facts}}\\
           &=\emptybag(a) \tag*{~\autoref{prop:merge-empty} \& \autoref{ass:simpl}}\\
           &=\filter_j(\Bag(\inj_i)(\fils{1}(\flat_i(F_i\chi_S(t)))))(a)
           \tag*{\autoref{rem:filter-facts}}\\
           \intertext{The remainder of the calculation is completely analogous
             to the first case.}
  \end{align*}
\item Since both the product and coproduct of functors share the same definition
  of \merge{}, its linear run-time complexity only needs to be verified once. To
  this end, we represent bags $\Bag(A)$ as (linked) lists of elements from $A$
  (in lieu of maps $A\to\N$) and rewrite the definition of \merge{} 
  such that it uses $\concat'$ from \eqref{defConcatPrime}:
  \begin{align*}
    \merge(t)(\inj_j a) &= \merge_j(\filter_j(t))(a) \tag*{\autoref{C:merge-prod}} \\
    &= \pr_j(\fpair{\merge_i\cdot\filter_i}_{i\in I}(t))(a) \tag*{Def.~of $\fpair{\cdots}$}\\
    &= \concat'(\fpair{\merge_i\cdot\filter_i}_{i \in I}(t))(\inj_j a)
      \tag*{\eqref{defConcatPrime}}
  \end{align*}
  Hence, $\merge$ is the composition
  \[
    \Bag\big(\coprodi A_i \big)
    \xra{\fpair{\filter_i}_{i\in I}}
     \prodi\Bag(A_i)
    \xra{\Pi_{i\in I}\;\merge_i}
     \prodi\Bag(A_i)
    \xra{\concat'}\Bag\big(\coprodi A_i\big).
  \]
  This composition can be readily implemented by the following algorithm. Given a
  bag $t\in \Bag\big(\coprodi A_i\big)$, let $n$ be the number of elements in $t$ and do:
  \begin{enumerate}
  \item Allocate an array of length $|I|$ initially containing an empty bag of
    type $\Bag(A_i)$ in the~$i$th component for all $i\in I$ (this array
    represents an element of $\prodi\Bag(A_i)$).
  \item Insert each label $\inj_i a$ from $t$ into the $i$th bag; this
    implements $\fpair{\filter_i}_{i\in I}$ above.
  \item For each $i\in I$, apply $\merge_i$ on the $i$th bag.
  \item Concatenate the resulting $|I|$ lists (encoding bags of type $\Bag(A_i)$)
    stored in our array to one list encoding the result bag of type $\Bag\big(\coprodi A_i\big)$.
  \end{enumerate}
  
  Each of those steps runs in $\CO(|I| + n)$ time if $\merge_i$ has linear
  run-time for every $i\in I$. Since $|I|$ is constant, this amounts to $\CO(n)$
  overall.\qedhere
\end{enumerate}
\end{proof}

\subsection{Proofs for \autoref{sec:reach}}

Gumm~\cite[Def.~7.2]{Gumm2005} defined the maps $\tau^F_X\colon FX \to \Pow X$
differently. We show that his definition is equivalent to ours. 

\begin{lemma}
  The definition of $\tau_X^F$ in \autoref{def:cangr} is equivalent to Gumm's definition
  in \opcit[].
\end{lemma}
\begin{proof}
  Before showing the equivalence, we need to recall other definitions that are
  used by Gumm~\cite{Gumm2005}. Recall that a filter $\mathcal{G}$ on a set $X$
  is a nonempty family $\mathcal{G}\subseteq \Pow(X)$ that is closed under
  binary intersection and supersets. The \emph{filter functor} $\mathbb{F}$ is
  the $\Set$-functor that sends a set $X$ to the set of all filters on $X$
  (its definition on maps is not relevant to this proof). For a given
  $\Set$-functor $F\colon \Set\to\Set$ and sets $U\subseteq X$, the set
  \[
    [F_U^X] := F(i\colon U\hookrightarrow X)[FU]  \qquad\subseteq FX
  \]
  denotes the image of $Fi\colon FU\to FX$, where $i\colon U\hookrightarrow X$
  is the inclusion map. The notation $[F_U^X]$ is monotone in $U$~\cite[Lemma 1]{Gumm2005}, that is,
  $V\subseteq U$ implies $[F_V^X]\subseteq [F_U^X]$. Moreover, one can easily
  prove for $t\in FX$
  that
  \begin{equation}
    t\in [F_U^X]
    \quad\Longleftrightarrow\quad \text{$t\colon 1\to FX$ factorizes through $Fi$}.
    \label{gummNotationVsFactor}
  \end{equation}
  Using this notation, we define the following family of maps:
  \[
    \mu_X\colon FX \to \mathbb{F}X
    \qquad
    \mu_X(t) := \big\{ U\subseteq X ~\big\vert~ t \in [F_U^X]\,\big\}.
  \]
  For $t\in FX$, the intersection of all elements in $\mu_X(t)$
  yields a subset of $X$:
  \[
    \bigcap \mu_X(t) =
    \{
    x\in X\mid \forall\, U\in \mu_X(t)\colon x\in U
    \}
 \quad\in \Pow X.
  \]
  This is the definition of $\tau_X^F$ in \opcit[]. In order to prove that this
  definition is equivalent to ours in \autoref{def:cangr}, we will prove that
  \[
    \tau_X^F(t) = \bigcap \mu_X(t)\qquad\text{for all }t\in FX.
  \]
  We have the following chain of equal sets, whose equality is established by
  performing equivalent rewrites in the comprehension formula:
  \begin{align*}
    \bigcap \mu_X(t)
    &= \{x\in X\mid \forall~ U\in \mu_X(t)\colon x\in U \} \tag{Def.~$\bigcap$}\\
    &= \{x\in X\mid \forall~ U\subseteq X\colon t\in [F_U^X] \to x\in U\} \tag{Def.~$\mu_X$} \\
    &= \{x\in X\mid \forall~ U\subseteq X\colon x\not\in U\to t\not\in [F_U^X]\} \tag{Contraposition} \\
    &= \{x\in X\mid \forall~ U\subseteq X\text{ with }x\not\in U\colon t\not\in [F_U^X]\} \\
    &= \{x\in X\mid \forall~ U\subseteq X\setminus\{x\}\colon t\not\in [F_U^X]\} \\
    &= \{x\in X\mid t\not\in [F_{X\setminus\{x\}}^X]\} \tag{$\forall$-instance \& $[F_U^X]$ monotone in $U$}\\
    &= \{x\in X\mid t\colon 1\to FX\text{ does not factor through }F(i\colon X\setminus\{x\}\to X)\}
      \tag*{\eqref{gummNotationVsFactor}} \\
    &= \tau_F^X(t)
    \tag*{Def.~$\tau_F^X(t)$}%
  \end{align*}
  This completes the proof.
\end{proof}

\subsection*{Proof of \autoref{prop:hin}}
\begin{proof}%
  We prove $\tau_{X}^{\Bag(A\times -)}(\flat_X(t))\subseteq \tau_X^F(t)$ by contraposition.
  If $x\in X$ is \emph{not} in $\tau_X^F(t)$, then
  we show that it is not in $\tau^{\Bag(A\times -)}_X(\flat_X(t))$ by proving
  that the following diagram commutes:
  \[
    \begin{tikzcd}
      1 \ar[r, "t"]\ar[dr, "t'", swap] & FX \ar[r, "\flat"] & \Bag(A\times X)\\
      & F(X\setminus\{x\})\ar[u, "Fi"]\ar[r, "\flat"] & \Bag(A\times
      (X \setminus \{x\})\ar[u, "\Bag(A\times i)",swap]
    \end{tikzcd}
  \]
  First, observe that $x \not\in \tau^F_X(t)$ implies that the map
  $t\colon 1\to FX$ factorizes through $F(X\setminus\{x\})\xrightarrow{Fi}FX$
  (cf.~\autoref{def:cangr}), and we therefore obtain $t'$ as shown in the left
  triangle. The right rectangle commutes by the subnaturality of
  $\flat$. Therefore, $1\xrightarrow{t}FX\xrightarrow{\flat}\Bag(A\times X)$
  factors through
  $\Bag(A\times X\setminus\{x\})\xrightarrow{\Bag(A\times i)}\Bag(A\times X)$
  and thus $x$ can not be in $\tau_X^{\Bag(A\times -)}(\flat(t))$. %
  We have shown that $x\not\in\tau_X^F(t)$ implies
  $x\not\in\tau_X^{\Bag(A\times -)}(\flat(t))$; equivalently, we have
  $\tau_X^{\Bag(A\times -)}(\flat(t))\subseteq\tau_X^F(t)$ as required.
\end{proof}

\subsection*{Proof of \autoref{prop:her}}
\begin{proof}%
  Having established one inclusion in \autoref{prop:hin} already, we prove the remaining inclusion
   $\tau_X^F(t) \subseteq \tau_{X}^{\Bag(A\times -)}(\flat_X(t))$ by contraposition:
  \[
    x\not \in \tau_{X}^{\Bag(A\times -)}(\flat_X(t))
    \Longrightarrow
    x\not\in \tau_X^F(t)
    \qquad
    \text{for all }x\in X.
  \]
  To this end, suppose that $x\in X$ satisfies
  $x\not \in \tau_{X}^{\Bag(A\times -)}(\flat_X(t))$. This implies
  that there exists some $t' \in \Bag(A\times (X\setminus\{x\}))$ such
  that the diagram below commutes:
  \[
    \begin{tikzcd}
      1
      \arrow{r}{t}
      \arrow{drr}[swap]{t'}
      & FX \arrow{r}{\flat_X}
      & \Bag(A\times X)
      \\
      && \Bag(A\times (X\setminus\{x\}))
      \arrow{u}[swap]{\Bag(A\times i)}
    \end{tikzcd}
    \text{ with }
    \begin{tikzcd}
      X
      \\
      X\setminus \{x\}.
      \arrow[>->]{u}[swap]{i}
    \end{tikzcd}
  \]
  For a fixed $b\in 2 = \{0,1\}$, we define the injective auxiliary map
  $(x\leadsto b)$ by
  \[
    (x\leadsto b)\colon X\to (X\setminus\{x\}) + 2
    \qquad
    (x\leadsto b)(y) = \begin{cases}
      b &\text{if }y = x \\
      y &\text{otherwise.}
    \end{cases}
  \]
  We will now prove the following equality
  \begin{equation}
    F(x\leadsto 0)(t) = F(x\leadsto 1)(t).
    \label{mapToAnything}
  \end{equation}
  We have the following commutative diagrams (for $b = 0,1$):
  \[
    \begin{tikzcd}[sep = 12mm, column sep=19mm]
      F(X\setminus\{x\}+2)
      \descto{dr}{$\fpair{\text{trivial},\text{subnaturality}}$}
      \arrow{r}{\fpair{F!, \flat_{X\setminus\{x\} + 2}}}
      & F1\times \Bag(A\times (X\setminus\{x\} + 2))
      \\
      FX
      \arrow{u}{F(x\leadsto b)}
      \arrow{r}{\fpair{F!, \flat_X}}
      & F1\times \Bag(A\times X)
      \arrow{u}[swap]{F1\times \Bag(A\times (x\leadsto b))}
      \\
      1
      \arrow{u}{t}
      \arrow{r}[swap]{\fpair{F!\,\cdot\, t,\, t'}}
      \descto{ur}{(Def.~of $t'$)}
      & F1\times \Bag(A\times (X\setminus\{x\}))
      \arrow{u}[swap]{F1\times \Bag(A\times i)}
      \arrow[shiftarr={xshift=28mm}]{uu}[swap]{F1\times \Bag(A\times \inl)}
    \end{tikzcd}
  \]
  Note that the lower right outside path via $t'$ does not mention
  $b$ at all. Hence, we have:
  \begin{align*}
    \fpair{F!, \flat_{X\setminus\{x\} + 2}}(F(x\leadsto 0)(t))
    &= (F1\times \Bag(A\times \inl))(F!(t), t')
    \\ &
    = \fpair{F!, \flat_{X\setminus\{x\} + 2}}(F(x\leadsto 1)(t)).
  \end{align*}
  By assumption $F$ has a refinement interface, and so we know that
  the map $\fpair{F!, \flat_{X\setminus\{x\} + 2}}$ is injective by
  \autoref{encoding-injective}.  Thus, we obtain the desired equation
  \eqref{mapToAnything}. Rephrased as a diagram, we see that the
  outside
  of the following diagram commutes:
  \[
    \begin{tikzcd}
      & FX
      \arrow{r}{F(x\leadsto 0)}
      & F(X\setminus\{x\} + 2)
      \\[3mm]
      & F(X\setminus\{x\})
      \arrow{u}{Fi}
      \arrow{r}{Fi}
      \pullbackangle{45}
      & FX
      \arrow{u}[swap]{F(x\leadsto 1)}
      \\
      1 \arrow[bend left]{uur}{t}
      \arrow[bend right=10]{urr}[swap]{t}
      \arrow[dashed]{ur}{t''}
    \end{tikzcd}
  \]
  Regarding the square, note that $X\setminus\{x\}$ is the
  intersection of the injective maps
  $(x\leadsto 0)\colon X\monoto X\setminus\{x\} + 2$ and
  $(x\leadsto 1)\colon X\monoto X\setminus\{x\} + 2$. Since $F$
  preserves intersections, we thus see that the above square is a
  pullback. Since the outside commutes, we obtain the above dashed map
  $t''\colon 1\to F(X\setminus\{x\})$ with $Fi\cdot t'' = t$. By the
  definition of $\tau_X^F$ (see~\autoref{def:cangr}), this implies that $x\not \in \tau_X^F(t)$,
  as desired.
\end{proof}